\documentclass[journal]{IEEEtran}

\usepackage{amsmath,amsfonts,amsthm,amssymb}
\usepackage{algorithmic}
\usepackage{algorithm}
\usepackage{array}
\usepackage[caption=false,font=normalsize,labelfont=sf,textfont=sf]{subfig}
\usepackage{textcomp}
\usepackage{stfloats}
\usepackage{verbatim}
\usepackage{graphicx}
\usepackage{cite}
\usepackage{hyperref}
\usepackage{mathtools}

\newtheorem{lemma}{Lemma}
\newtheorem{proposition}{Proposition}
\newtheorem{theorem}{Theorem}
\newtheorem{remark}{Remark}

\begin{document}

\title{\LARGE Towards Autonomous Driving with Short-Packet Rate Splitting:\\Age of Information Analysis and Optimization}

\author{
Zirui Zheng, Yingyang Chen,~\IEEEmembership{Senior Member,~IEEE}, Xinyue Pei,~\IEEEmembership{Member,~IEEE}, Xingwei Wang,\\Zhiquan Liu,~\IEEEmembership{Senior Member,~IEEE},~Theodoros A. Tsiftsis,~\IEEEmembership{Senior Member,~IEEE},\\Miaowen Wen,~\IEEEmembership{Senior Member,~IEEE},~and Pingzhi Fan,~\IEEEmembership{Fellow,~IEEE}
\thanks{Z. Zheng and Y. Chen are with the College of Information Science and Technology, Jinan University, Guangzhou 510632, China (e-mail: ruszzr@stu2024.jnu.edu.cn; chenyy@jnu.edu.cn).}
\thanks{X. Pei and X. Wang are with the School of Computer Science and Engineering, Northeastern University, Shenyang 110819, China (e-mail: peixy@cse.neu.edu.cn; wangxw@mail.neu.edu.cn).}
\thanks{Z. Liu is with the College of Cyber Security, Jinan University, Guangzhou 510632, China (e-mail: zqliu@vip.qq.com).}
\thanks{T. A. Tsiftsis is with the Department of Informatics and Telecommunications, University of Thessaly, Lamia 35100, Greece (e-mail: theodoros.tsiftsis@nottingham.edu.cn; tsiftsis@uth.gr).}
\thanks{M. Wen is with the School of Electronic and Information Engineering, South China University of Technology, Guangzhou 510640, China (email: eemwwen@scut.edu.cn).}
\thanks{P. Fan is with the Key Laboratory of Information Coding and Transmission, Southwest Jiaotong University, Chengdu 610031, China (e-mail: p.fan@ieee.org).}
}
\maketitle

\begin{abstract}
To address the high mobility impacts and the ultra-reliable and low-latency communication (URLLC) requirements in autonomous driving scenarios, rate-splitting multiple access (RSMA) combined with short-packet communication (SPC) emerges as a promising solution. Autonomous vehicles rely on real-time information exchange to ensure safety and coordination, making information freshness essential. By jointly capturing transmission delays and packet errors, age of information (AoI) serves as a comprehensive metric for freshness. In this paper, we investigate short-packet rate splitting to enhance information freshness measured by the AoI. By splitting the unicast messages into common and private parts, encoding all common parts together with the multicast message into a common stream, and encoding each private part into a private stream, RSMA effectively manages interference and enables achieving lower AoI. By considering critical factors such as transmit power, vehicle velocity, blocklength, and the number of transmit antennas, we derive closed-form expressions for the average AoI (AAoI) of the common stream under partial decoding and the overall AAoI under complete decoding. To enhance the AAoI performance, we propose the multi-start two-step successive convex approximation (SCA) algorithm. This algorithm first optimizes the power allocation and subsequently optimizes the rate splitting under the quality of service (QoS) trade-off constraint. Simulation results demonstrate that our short-packet rate-splitting scheme significantly improves the AAoI performance while ensuring system fairness and enabling ultra-low AAoI through the common stream, meeting the requirements of autonomous driving applications. Moreover, the trade-off between the common and overall performance is revealed, indicating that the overall performance can be further enhanced while maintaining the advantages of the common stream.
\end{abstract}


\section{Introduction}
Autonomous driving is defined as a technology that enables vehicles to accurately perceive their surrounding environment and operate reliably without human intervention~\cite{AD}. It contributes to safer transportation and more sustainable travel, while significantly reducing accidents and emissions. To support such applications, the emerging sixth-generation (6G) vehicular-to-everything (V2X) communications provide ultra-reliable and low-latency communication (URLLC), which is essential for real-time sensing, control, and coordination between vehicles and infrastructure~\cite{6G-V2X}. 

As a time-critical application, autonomous driving requires real-time information exchange to ensure coordinated control and safe maneuvering. Stale status updates may lead to incorrect control decisions and compromise driving safety, thereby imposing stringent requirements on information freshness. For example, service scenarios such as vehicle platooning and cooperative perception demand end-to-end latency below tens of milliseconds, and advanced services like remote driving even require latency less than $5$~ms~\cite{3gpp.22.186}. However, transmission latency merely measures the delivery time of successfully received packets, overlooking the cumulative deterioration of freshness induced by reception errors. Consequently, conventional metrics, e.g., transmission latency, are inadequate for capturing the freshness performance of continuously evolving systems~\cite{needAoI}. To address this, the age of information (AoI) is introduced as a metric for measuring information freshness~\cite{AoIintroduce}. It is defined as the elapsed time since the generation of the latest successfully received packet, with smaller values indicating fresher status updates~\cite{AoIdefine}. By jointly accounting for transmission delays and packet errors, AoI provides a comprehensive measure of the freshness in wireless networks~\cite{AoImeasure1}.

To satisfy the stringent latency demands imposed by URLLC, short-packet communication (SPC) with finite blocklength (FBL) coding is essential~\cite{FBL}. Unlike the conventional infinite blocklength assumption, FBL coding reduces the latency while inevitably introducing transmission errors, which adversely affect the information freshness. In addition, autonomous vehicles operate in highly dynamic and complex environments, where the high mobility introduces Doppler shifts and delays, making it difficult to obtain perfect channel state information (CSI) at the transmitter (CSIT)~\cite{CSIT}. Consequently, the assumption of perfect CSIT becomes impractical, and the impact of CSIT errors cannot be neglected. Both packet errors induced by FBL coding and degraded communication performance caused by imperfect CSIT can significantly affect the AoI performance. Hence, it is necessary to study the AoI performance under FBL transmission and imperfect CSIT.

Given the challenges posed by FBL transmission and imperfect CSIT, a reliable multiple access scheme becomes necessary. Rate-splitting multiple access (RSMA) was first proposed in~\cite{firstRSMA}, demonstrating that RSMA enables Gaussian multiple access channels to achieve any capacity region point with low complexity and without user synchronization. Building on this foundation, RSMA has evolved into a promising and practical multiple access scheme tailored to modern wireless communications. By performing rate splitting at the transmitter and successive interference cancellation (SIC) at the receivers~\cite{mechanism_independent}, RSMA achieves a flexible interference management strategy that partially treats interference as noise and partially decodes it. Consequently, it avoids the extremes of space division multiple access (SDMA), which treats all interference as noise, and non-orthogonal multiple access (NOMA), which fully decodes it, serving as a soft bridge between the two~\cite{bridging_SIC}. By effectively managing multi-user interference, RSMA exhibits strong robustness and flexibility, making it particularly suitable for high-mobility SPC scenarios~\cite{mobility2}.

\subsection{Related Work}
Recently, RSMA has been extensively investigated, proving its superiority over conventional schemes. Specifically, \cite{dirtypaper} shows that RSMA achieves a broader rate region than dirty paper coding in multi-antenna broadcast channels under partial CSIT. Regarding the maximization of the ergodic sum-rate, \cite{mechanism_independent} derived a closed-form power allocation under imperfect CSIT, and \cite{FBL_perfectCSIT} designed the optimal linear precoders under FBL transmission. Extending these works, \cite{FBL_imperfectCSIT} derived a closed-form expression for the sum-rate under imperfect CSIT and FBL. Building on this, the work in \cite{wangyi} jointly optimized power allocation and common rate splitting to further enhance the ergodic performance. Moreover, \cite{throughput_SICerror,throughput_SICflexible} demonstrate that RSMA achieves higher effective throughput under FBL. In terms of reliability, \cite{error_URLLC} verifies that RSMA effectively reduces the error probability under FBL. Finally, \cite{fair_uplink,fair_downlink} show that RSMA ensures excellent user fairness in uplink and downlink FBL transmission. These findings collectively demonstrate the effectiveness of RSMA in addressing user interference and reception errors, highlighting it as a promising multiple access paradigm for next-generation wireless networks.

Despite growing interest in RSMA, research on its information freshness remains limited, and only a few works have begun to explore AoI performance across different scenarios. The studies in~\cite{eMBB_AoI,IoT_AoI,semantic_AoI} are based on the idealized assumptions of perfect CSIT and infinite blocklength. Specifically, \cite{eMBB_AoI} investigated the coexistence of mission-critical and enhanced mobile broadband services, evaluating the average AoI (AAoI) and peak AoI violation probability for puncturing, NOMA, and RSMA. The results show that RSMA can achieve comparable AoI performance while significantly improving data rates. In Internet of Things (IoT) networks, \cite{IoT_AoI} demonstrates that RSMA effectively reduces AAoI compared with orthogonal multiple access and NOMA. The work in~\cite{semantic_AoI} introduced RSMA into semantic communications and shows that it achieves a lower age of incorrect information (AoII) than SDMA. To better align with practical communication scenarios, \cite{HARQ_AoI,sb_downlink_AAoI,sb_uplink_AoII} extended the research to FBL transmission. In~\cite{HARQ_AoI}, a joint framework of hybrid automatic repeat request and RSMA was proposed to minimize AAoI by adaptive power allocation and retransmissions. In satellite-based IoT systems, \cite{sb_downlink_AAoI} and~\cite{sb_uplink_AoII} evaluated the AAoI and AoII, respectively, demonstrating the superior AoI performance of RSMA in SPC. However, these works assume perfect CSIT and do not account for the inevitable CSIT errors in high-mobility scenarios, which significantly degrade AoI performance. Moreover, the unique advantages of the rate-splitting mechanism in optimizing AoI have yet to be fully leveraged. Therefore, investigating the AoI performance of RSMA by addressing these limitations remains a significant research gap.

\subsection{Motivation and Contributions}
Existing studies have demonstrated the superior performance of RSMA across various communication scenarios. However, its potential for information freshness has not been fully explored, particularly under imperfect CSIT and FBL transmission. RSMA can deliver multicast messages through the common stream to avoid redundant transmissions and flexibly split unicast messages into common and private parts for efficient interference management. Therefore, an appropriate rate-splitting strategy can potentially mitigate the AoI degradation caused by imperfect CSIT and large-scale fading. Furthermore, since the common stream is decoded at the first layer of SIC, it can still be successfully received in partial decoding scenarios, where the common stream is decoded successfully while the private stream fails to decode. Consequently, the common stream achieves a lower AoI than the private streams, allowing unicast messages to be assigned to it for higher freshness. Motivated by these observations, we investigate the AoI performance of short-packet rate splitting in autonomous driving scenarios. The main contributions of this paper are summarized as follows:
\begin{itemize}
\item We establish a framework to analyze and optimize the AoI performance of short-packet rate splitting in high-mobility autonomous driving scenarios. Specifically, we analyze the common performance under partial decoding and the overall performance under complete decoding, and derive the closed-form expressions for the AAoI. The continuous nature of these expressions provides an analytical foundation for subsequent AAoI optimization.
\item We propose an algorithm to minimize the average AAoI of vehicles. Specifically, we initialize multiple sets of rate-splitting factors as starting points. For each starting point, we first optimize the power allocation, then enable the quality of service (QoS) trade-off constraint to optimize the rate splitting, with both subproblems solved via successive convex approximation (SCA). Finally, the optimal solution is selected from all starting points. The algorithm effectively addresses the strong coupling between power allocation and rate-splitting factors, enabling an efficient solution to the optimization problem.
\item We verify the validity of the closed-form expressions and the convergence of the optimization algorithm, and demonstrate that the proposed RSMA scheme outperforms SDMA and NOMA. Specifically, the proposed scheme effectively reduces the average AAoI of vehicles while ensuring user fairness and delivers significantly lower AAoI through the common stream. In addition, we investigate the trade-off between the common and overall performance induced by the QoS constraint, indicating that appropriate parameters can enhance overall performance while maintaining the common stream advantage.
\end{itemize}

The rest of the paper is organized as follows. Section~\ref{section II} presents the system model. Section~\ref{section III} derives the closed-form AAoI expressions for both the common stream and the overall system. Section~\ref{section IV} formulates the problem to minimize the average AAoI of vehicles and uses the proposed algorithm to optimize the power allocation and rate splitting. Section~\ref{section V} provides the simulation results and performance analysis, and Section~\ref{section VI} concludes the paper.

\textit{Notation}: Vectors are represented by bold lowercase letters. $\vert \cdot \vert$ and $\| \cdot \|$ indicate absolute value of a scalar and $\ell_2$-norm of a vector, respectively. $(\cdot)^T$ and $(\cdot)^H$ indicate transpose and Hermitian transpose of vectors, respectively. $\mathbf{I}$ denotes the identity matrix. $\mathbb{C}$ refers to the set of complex numbers. $\mathcal{CN}(0,\sigma^2)$ denotes the circularly symmetric complex Gaussian distribution with zero mean and variance $\sigma^2$. $\mathrm{Gamma}(D,\theta)$ denotes distribution with probability density function $f(x) = \frac{x^{D-1}e^{-x/\theta}}{\Gamma(D)\theta^D}$, with shape parameter $D$ and scale parameter $\theta$, and $\Gamma(D)=\int_0^\infty t^{D-1}e^{-t} \, \mathrm{d}t$ is the gamma function. $\ln(\cdot)$ denotes the natural logarithm. $\mathbb{P}(\cdot)$ and $\mathbb{E}[\cdot]$ represent probability and expectation of random variables, respectively.

\section{System Model}
\label{section II}
\begin{figure}[t] 
\centering 
\includegraphics[width=0.95\linewidth]{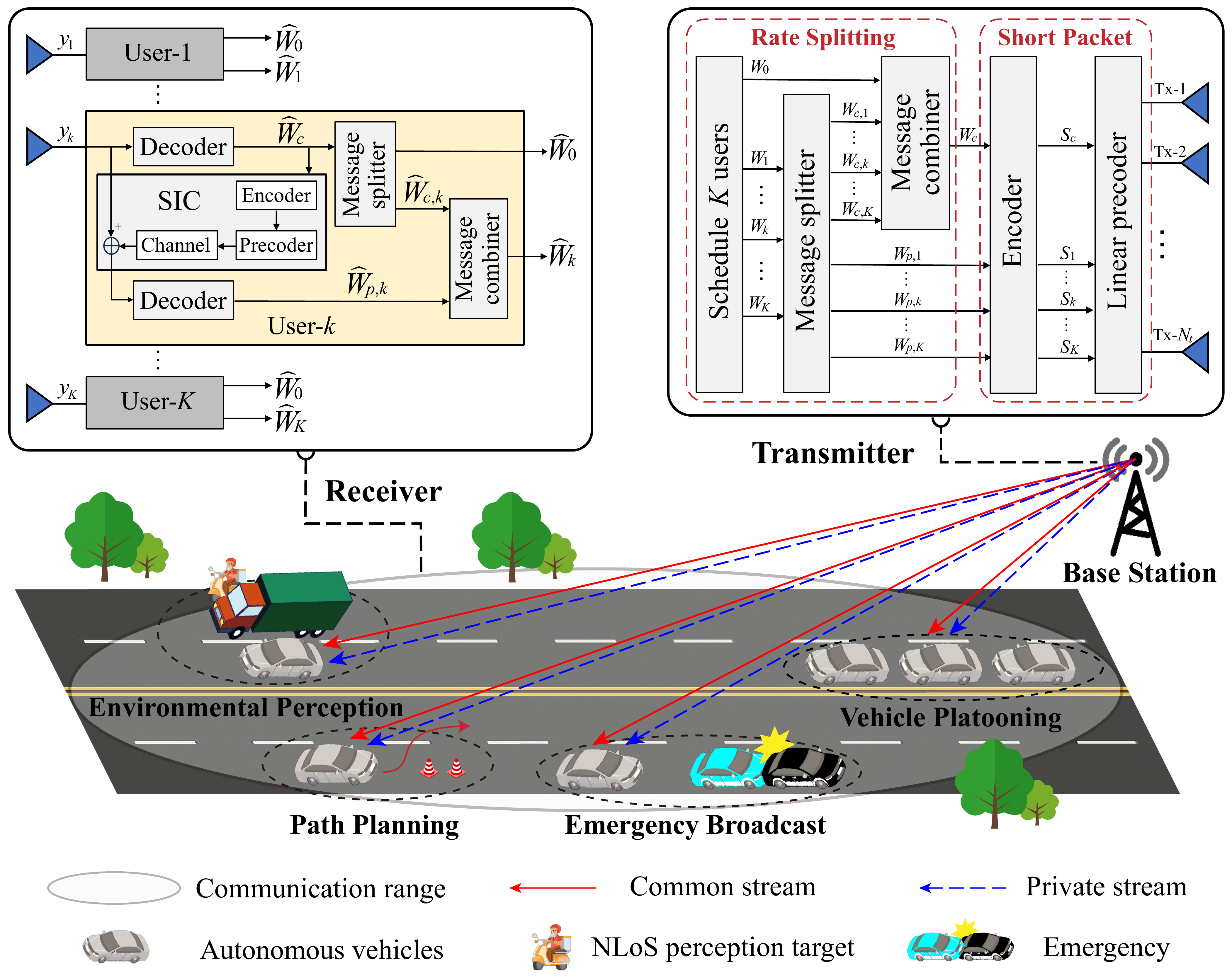}
\vspace{-5pt}
\caption{Short-packet rate splitting in a high-mobility autonomous driving scenario.}
\vspace{-10pt}
\label{fig1} 
\end{figure}
Fig.~\ref{fig1} depicts a typical high-mobility autonomous driving communication scenario with short-packet rate splitting. The BS is equipped with $N_t$ antennas in a downlink multiple-input single-output (MISO) broadcast setup and serves $K$ single-antenna autonomous vehicles, where $N_t \geq K$. The BS transmits multicast messages, such as vehicle platooning~\cite{multicast-platoon}, environmental perception~\cite{multicast-perception}, and emergency broadcasts~\cite{multicast-emergency}, as well as unicast messages, including path planning and remote driving~\cite{unicast}. The vehicles are indexed by $\mathcal{K} = \{1,2,\ldots,K\}$. A single-layer RSMA strategy is adopted to reduce receiver complexity, where each vehicle performs single-layer SIC decoding for the received signal~\cite{bridging_SIC}. At the BS transmitter, the multicast message for all vehicles is denoted as $W_0$, and the unicast message for vehicle-$k$ is denoted as $W_k$, where $\forall k \in \mathcal{K}$. Each unicast message $W_k$ is split into the common part $W_{c,k}$ and the private part $W_{p,k}$. The multicast message $W_0$ and all common parts $\{W_{c,1},W_{c,2},\ldots,W_{c,K}\}$ are combined into the common message $W_c$, which is then encoded into the common stream $s_c$~\cite{supercommon}. Each private part $W_{p,k}$ is independently encoded into the private stream $s_k$. Both the common and private streams are encoded under the FBL regime, which introduces non-negligible packet errors. At the receiver, each vehicle first decodes the common stream, then decodes its own private stream via SIC. Linear precoding is applied to all streams. The composite stream vector is denoted as $\mathbf{s}=[s_c,s_1,s_2,\ldots,s_K]^T$ and satisfies $\mathbb{E}[\mathbf{ss}^H] = \mathbf{I}$. The transmitted signal at the BS is given by
\begin{equation}
\mathbf{x} = \sqrt{P_T\alpha_c}\mathbf{p}_cs_c + \sqrt{P_T}\sum_{k\in\mathcal{K}}\sqrt{\alpha_k}\mathbf{p}_ks_k,
\end{equation}
where $\mathbf{p}_c,\mathbf{p}_k \in \mathbb{C}^{N_t\times1}$ denote the precoders for the common and the $k$-th private streams, respectively, satisfying $\| \mathbf{p}_c \|^2 = \| \mathbf{p}_k \|^2 = 1, \forall k \in \mathcal{K}$. The total transmit power is represented by $P_T$. The power allocation factors for the common and the $k$-th private streams are represented by $\alpha_c,\alpha_k \in [0,1]$, respectively, with $\alpha_c + \sum_{k \in \mathcal{K}}\alpha_k = 1$. For tractable analysis, zero-forcing (ZF) precoding is adopted for the private streams, and random precoding is employed for the common stream~\cite{mechanism_independent,wangyi}.

The received signal at vehicle-$k$ is given by
\begin{equation}
y_k=\sqrt{\xi_k}\mathbf{h}_k^H[m]\mathbf{x}+n_k,
\end{equation}
where $\mathbf{h}_k[m] \in \mathbb{C}^{N_t\times1}$ denotes the small-scale fading between the BS and vehicle-$k$, $n_k \sim \mathcal{CN}(0,\sigma^2)$ represents the additive Gaussian white noise (AWGN), and $\xi_k$ indicates the power attenuation due to the path loss.

We consider that both vehicle mobility and CSI feedback delay lead to outdated CSIT, and employ the first-order Gaussian-Markov process to model the channel evolution~\cite{channel}:
\begin{equation}
\mathbf{h}_k[m] = \sqrt{\rho^2}\mathbf{h}_k[m-1] + \sqrt{1-\rho^2}\mathbf{e}_k[m],
\end{equation}
where $\mathbf{h}_k[\cdot]$ is a spatially uncorrelated Rayleigh flat fading channel with independent and identically distributed (i.i.d.) entries following $\mathcal{CN}(0,1)$. Similarly, the error term $\mathbf{e}_k[m]$ consists of i.i.d. entries following $\mathcal{CN}(0,1)$. The time correlation coefficient $\rho=J_0(2\pi f_DT)$ characterizes the temporal dependence between $\mathbf{h}_k[m]$ and $\mathbf{h}_k[m-1]$, obeying the Jake's model~\cite{jakes}, where $J_0(\cdot)$ is the zeroth-order Bessel function. Here, $T$ denotes the channel instantiation interval, and $f_D = vf_c/c$ represents the maximum Doppler frequency, where $f_c$ is the carrier frequency, $v$ is the vehicle velocity, and $c = 3 \times 10^8$ m/s is the speed of light. The model assumes that the channel remains constant within each coherence interval, and the transmitter can only obtain the outdated CSIT. Specifically, the actual channel vector at time slot $m$ is $\mathbf{h}_k[m]$, while the transmitter only obtains $\mathbf{h}_k[m-1]$ from the previous slot for precoding. Additionally, the path loss (in~dB) is given by $PL_k = -32.4 - 20\lg(f_c) - 31.9\lg(d_k)$~\cite{3gpp.38.901}, and the linear-scale power attenuation factor is expressed as $\xi_k=10^{PL_k/10}$.

\section{Closed-Form AAoI Analysis}
\label{section III}
This section is devoted to deriving the closed-form AAoI expressions for both the common stream and the overall system. Firstly, the approximate cumulative distribution functions (CDFs) of the signal-to-interference-plus-noise ratios (SINRs) are obtained by leveraging the distributions of commonly used random variables. Subsequently, the approximate block error rates (BLERs) are derived by further applying the linear approximation of the Q-function. Finally, based on the AoI analysis, the AAoI expressions are obtained for use in the optimization discussed in the next section.

\subsection{CDFs of SINRs}
Upon receiving the signal, each vehicle first decodes the common stream by treating the private streams as noise. Thus, the common stream SINR at vehicle-$k$ is given by
\begin{equation}
\Gamma_{c,k} = \frac{P_{n}\xi_k\alpha_c \vert \mathbf{h}_k^H[m]\mathbf{p}_c \vert^2}{P_{n}\xi_k \sum_{j\in\mathcal{K}}\alpha_j \vert \mathbf{h}_k^H[m]\mathbf{p}_j \vert^2 + 1},
\end{equation}
where $P_{n} = P_T/\sigma^2$ denotes the normalized transmit power. After successfully decoding the common stream, we assume perfect SIC, meaning the common stream can be completely removed from the received signal. Subsequently, vehicle-$k$ decodes its private stream by treating other private streams as noise, with the SINR given by
\begin{equation}
\Gamma_{p,k} = \frac{P_{n}\xi_k\alpha_k \vert \mathbf{h}_k^H[m]\mathbf{p}_k \vert^2}{P_{n}\xi_k \sum_{j\in\mathcal{K}\setminus\{k\}}\alpha_j \vert \mathbf{h}_k^H[m]\mathbf{p}_j \vert^2+1}.
\end{equation}

The private streams perform ZF precoding based on the outdated CSIT from the previous time slot. Hence, for the $k$-th private stream, we have $\vert\mathbf{h}_k^H[m-1]\mathbf{p}_j\vert = 0$ for all $j \in \mathcal{K} \setminus \{k\}$, and $\vert\mathbf{h}_k^H[m-1]\mathbf{p}_k\vert^2 \sim \mathrm{Gamma}(N_t-K+1,1)$~\cite{Gamma}. The ZF precoder $\mathbf{p}_k$ is isotropically distributed and independent of the Gaussian error $\mathbf{e}_j[m]$ for all $j \in \mathcal{K}$, resulting in $\vert \mathbf{e}_j^H[m]\mathbf{p}_k\vert^2 \sim \mathrm{Gamma}(1,1)$. Furthermore, the common precoder $\mathbf{p}_c$ is assumed to be a random beamformer independent of $\mathbf{h}_k[m-1]$, $\mathbf{e}_k[m]$, and $\mathbf{p}_k$, leading to $\vert\mathbf{h}_k^H[m]\mathbf{p}_c\vert^2 \sim \mathrm{Gamma}(1,1)$~\cite{hkpc}. For tractability in deriving the statistics of the received SINRs, we neglect terms containing both $\mathbf{h}_k^H[m-1]$ and $\mathbf{e}_k^H[m]$, thereby approximating $\vert \mathbf{h}_k^H[m]\mathbf{p}_j\vert^2$ for all $j \in \mathcal{K}$ as
\begin{align}
\vert \mathbf{h}_k^H[m]\mathbf{p}_k\vert^2 &\approx \rho^2\vert \mathbf{h}_k^H[m-1]\mathbf{p}_k\vert^2 + (1\!-\!\rho^2)\vert \mathbf{e}_k^H[m]\mathbf{p}_k\vert^2, \label{6} \\
\vert \mathbf{h}_k^H[m]\mathbf{p}_j\vert^2 &\approx \rho^2\vert \mathbf{h}_k^H[m-1]\mathbf{p}_j\vert^2 + (1\!-\!\rho^2)\vert \mathbf{e}_k^H[m]\mathbf{p}_j\vert^2 \nonumber \\ 
&= (1-\rho^2)\vert \mathbf{e}_k^H[m]\mathbf{p}_j\vert^2,\; j \neq k. \label{7}
\end{align}

The SINR expressions involve sums of random variables (RVs) that follow different Gamma distributions. By matching the first two moments, the sum of Gamma-distributed RVs can be approximated by a new Gamma-distributed RV~\cite{Gamma}. Specifically, an RV $S = \sum_iA_i$, where $A_i \sim \mathrm{Gamma}(D_i,\theta_i)$, can be approximated by an RV $\widetilde{S} \sim \mathrm{Gamma}(\widetilde{D},\widetilde{\theta})$, where
\begin{equation}
\label{8}
\widetilde{D}=\frac{(\sum_iD_i\theta_i)^2}{\sum_i D_i\theta_i^2}, \; \widetilde{\theta}=\frac{\sum_iD_i\theta_i^2}{\sum_iD_i\theta_i}.
\end{equation}
\begin{lemma}
\label{lemma1}
The RV $X = \vert \mathbf{h}_k^H[m]\mathbf{p}_k\vert^2$ can be approximated by an RV $\widetilde{X}$ with the distribution $\mathrm{Gamma}(\widetilde{D}_1,\widetilde{\theta}_1)$, where
\begin{align}
\label{9}
\widetilde{D}_1 &= \frac{\left[(N_t-K+1)\rho^2+1-\rho^2\right]^2}{(N_t-K+1)\rho^4+(1-\rho^2)^2},\nonumber\\
\widetilde{\theta}_1 &= \frac{(N_t-K+1)\rho^4+(1-\rho^2)^2}{(N_t-K+1)\rho^2+1-\rho^2}.
\end{align}
\begin{proof}
See Appendix~\ref{appendixA}.
\end{proof}
\end{lemma}
\begin{lemma}
\label{lemma2}
The RV $Y_k = \sum_{j\in\mathcal{K}}\alpha_j\vert \mathbf{h}_k^H[m]\mathbf{p}_j\vert^2$ can be approximated by an RV $\widetilde{Y}_k \sim \mathrm{Gamma}(\widetilde{D}_{2,k},\widetilde{\theta}_{2,k})$, where
\begin{align}
\label{10}
\widetilde{D}_{2,k} &= \frac{\left[(N_t-K+1)\rho^2\alpha_k+(1-\rho^2)(1-\alpha_c)\right]^2}{(N_t-K+1)\rho^4\alpha_k^2+(1-\rho^2)^2\sum_{j\in\mathcal{K}}\alpha_j^2},\nonumber\\
\widetilde{\theta}_{2,k} &= \frac{(N_t-K+1)\rho^4\alpha_k^2+(1-\rho^2)^2\sum_{j\in\mathcal{K}}\alpha_j^2}{(N_t-K+1)\rho^2\alpha_k+(1-\rho^2)(1-\alpha_c)}.
\end{align}
\end{lemma}
\begin{lemma}
\label{lemma3}
The RV $Z_k = \sum_{j\in\mathcal{K}\setminus\{k\}}\alpha_j\vert \mathbf{h}_k^H[m]\mathbf{p}_j\vert^2$ can be approximated by an RV $\widetilde{Z}_k \sim \mathrm{Gamma}(\widetilde{D}_{3,k},\widetilde{\theta}_{3,k})$, where
\begin{equation}
\widetilde{D}_{3,k} \!=\! \frac{(1\!-\!\alpha_c\!-\!\alpha_k)^2}{\sum_{j\in\mathcal{K}\setminus\{k\}}\!\alpha_j^2}, \;
\widetilde{\theta}_{3,k} \!=\! \frac{(1\!-\!\rho^2)\sum_{j\in\mathcal{K}\setminus\{k\}}\!\alpha_j^2}{1\!-\!\alpha_c\!-\!\alpha_k}.
\end{equation}
\begin{proof}
The proofs of Lemmas~\ref{lemma2} and~\ref{lemma3} are similar to that of Lemma~\ref{lemma1} and are omitted here for brevity.
\end{proof}
\end{lemma}
Based on the approximate distributions of the above RVs, the approximate CDF of $\Gamma_{c,k}$ is derived in Lemma~\ref{lemma4}.
\begin{lemma}
\label{lemma4}
$\Gamma_{c,k}$ can be approximated by $\widetilde{\Gamma}_{c,k}$ with the CDF
\begin{equation}
\label{12}
F_{\widetilde{\Gamma}_{c,k}}(x) = 1 - \frac{e^{-\tfrac{x}{P_n\xi_k\alpha_c}}}{\Bigl( \frac{\widetilde{\theta}_{2,k}}{\alpha_c}x+1 \Bigr)^{\widetilde{D}_{2,k}}}, \; x\geq0.
\end{equation}
\begin{proof}
See Appendix~\ref{appendixB}.
\end{proof}
\end{lemma}
Lemma~\ref{lemma5} presents the approximate CDF of $\Gamma_{p,k}$. Since the derivation is similar to Lemma~\ref{lemma4}, the proof is omitted here.
\begin{lemma}
\label{lemma5}
$\Gamma_{p,k}$ can be approximated by $\widetilde{\Gamma}_{p,k}$ with the CDF
\begin{equation}
F_{\widetilde{\Gamma}_{p,k}}(x) = 1-\frac{e^{-\tfrac{x}{P_n\xi_k\widetilde{D}_1\widetilde{\theta}_1\alpha_k}}}{\Bigl( \frac{\widetilde{\theta}_{3,k}}{\widetilde{D}_1\widetilde{\theta}_1\alpha_k}x+1 \Bigr)^{\widetilde{D}_{3,k}}}, \; x\geq0.
\end{equation}
\end{lemma}

\subsection{BLER Analysis}
In FBL transmission, the BLER can be approximated as~\cite{polyanskiy}
\begin{equation}
\label{14}
\varepsilon \approx Q\left( \frac{\sqrt{n}\left[\log_2(1+\Gamma)-\frac{m}{n}\right]}{\log_2e\sqrt{1-(1+\Gamma)^{-2}}} \right),
\end{equation}
where $Q(x) = \frac{1}{\sqrt{2\pi}} \int_x^{\infty}e^{-\frac{t^2}{2}} \, \mathrm{d}t$ is the Gaussian Q-function, $\Gamma$ is the SINR, $n$ is the blocklength, and $m$ is the number of information bits. Due to the complexity of (\ref{14}), it is difficult to directly obtain a closed-form expression. Therefore, we approximate it using a piecewise linear function~\cite{Qfunction}, which can be expressed as
\begin{equation}
\label{15}
L(\Gamma) = 
\begin{cases}
1 &0\leq\Gamma\leq\mu \\
0.5-\delta(\Gamma-\beta) &\mu\leq\Gamma\leq\nu \\
0 &\Gamma\geq\nu
\end{cases},
\end{equation}
where $\delta = \sqrt{n/[2\pi(2^{2m/n}-1)]}$, $\beta=2^{\frac{m}{n}}-1$, $\mu=\beta-\frac{1}{2\delta}$, and $\nu =\!\beta+\frac{1}{2\delta}$. Based on the above approximation $\varepsilon \approx L(\Gamma)$ and the approximate CDFs of the SINRs, the approximate average BLERs are derived in Propositions~\ref{proposition1} and~\ref{proposition2}.
\begin{proposition}
\label{proposition1}
The average BLER for decoding the common stream at vehicle-$k$ can be approximated by
\begin{equation}
\label{16}
\widetilde{\varepsilon}_{c,k} = 1 - \frac{e^{-\tfrac{2^{\frac{m_c}{n_c}}-1}{P_n\xi_k\alpha_c}}} {\Bigl (\frac{\widetilde{\theta}_{2,k}}{\alpha_c} (2^{\frac{m_c}{n_c}}-1)+1 \Bigl)^{\widetilde{D}_{2,k}}}.
\end{equation}
\end{proposition}
\begin{proof}
See Appendix~\ref{appendixC}.
\end{proof}
\begin{proposition}
\label{proposition2}
Under perfect SIC and successful decoding of the common stream, the average BLER for decoding the private stream at vehicle-$k$ can be approximated by
\begin{equation}
\label{17}
\widetilde{\varepsilon}_{p,k} = 1 - \frac{e^{-\tfrac{2^{\frac{m_k}{n_k}}-1}{P_n\xi_k\widetilde{D}_1\widetilde{\theta}_1\alpha_k}}}{\Bigl( \frac{\widetilde{\theta}_{3,k}}{\widetilde{D}_1\widetilde{\theta}_1\alpha_k}(2^{\frac{m_k}{n_k}}-1)+1 \Bigr)^{\widetilde{D}_{3,k}}}.
\end{equation}
\end{proposition}
\begin{proof}
The proof is similar to that of Proposition~\ref{proposition1} and is omitted here.
\end{proof}
Each vehicle first decodes the common stream and then decodes its private stream via SIC. If the common stream decoding fails, SIC cannot remove its interference. Since the common stream is generally transmitted with higher power, it strongly interferes with the private streams, leading to decoding failure. Therefore, we assume that the common stream decoding failure inevitably results in the corresponding private stream decoding failure. Accordingly, the overall average BLER at vehicle-$k$ can be approximated by
\begin{equation}
\label{18}
\widetilde{\varepsilon}_k = \widetilde{\varepsilon}_{c,k} + (1-\widetilde{\varepsilon}_{c,k})\widetilde{\varepsilon}_{p,k}.
\end{equation}

\subsection{AAoI Analysis}
\begin{figure}[t] 
\centering 
\includegraphics[width=0.95\linewidth]{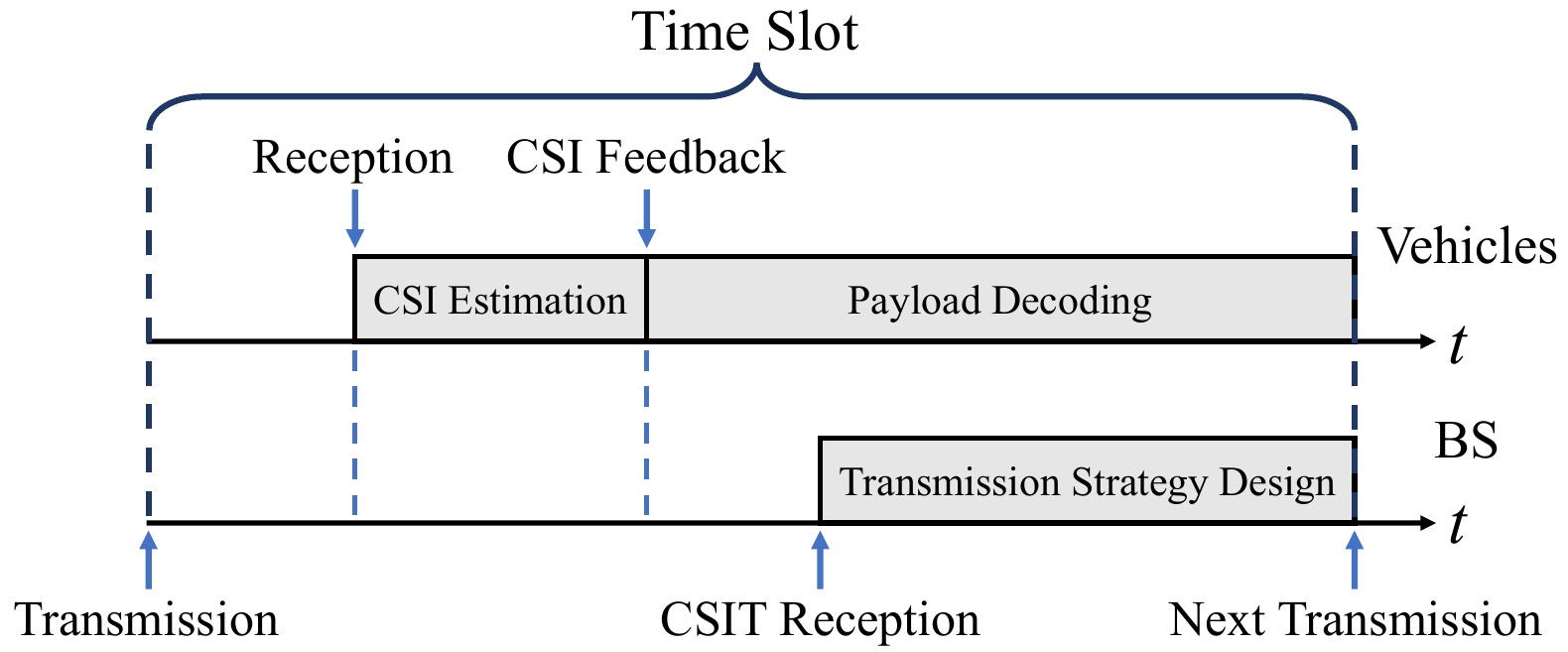}
\vspace{-5pt}
\caption{Simplified transmission model.}
\vspace{-10pt}
\label{fig2} 
\end{figure}
We adopt the Gauss–Markov channel model to characterize the outdated CSIT and assume each time slot accommodates one complete transmission and reception, as shown in Fig.~\ref{fig2}. At the beginning of a slot, the BS transmits pilots and data payloads. Vehicles first decode the pilots to estimate CSI and feed it back, then decode the payloads using SIC. Based on the acquired CSIT, the BS performs resource scheduling, encoding, and precoding for the next transmission. At the beginning of the next slot, vehicles complete decoding while the BS transmits the subsequent signal. This simplified transmission model facilitates a tractable analysis of AoI.

\begin{figure}[t] 
\centering 
\includegraphics[width=0.95\linewidth]{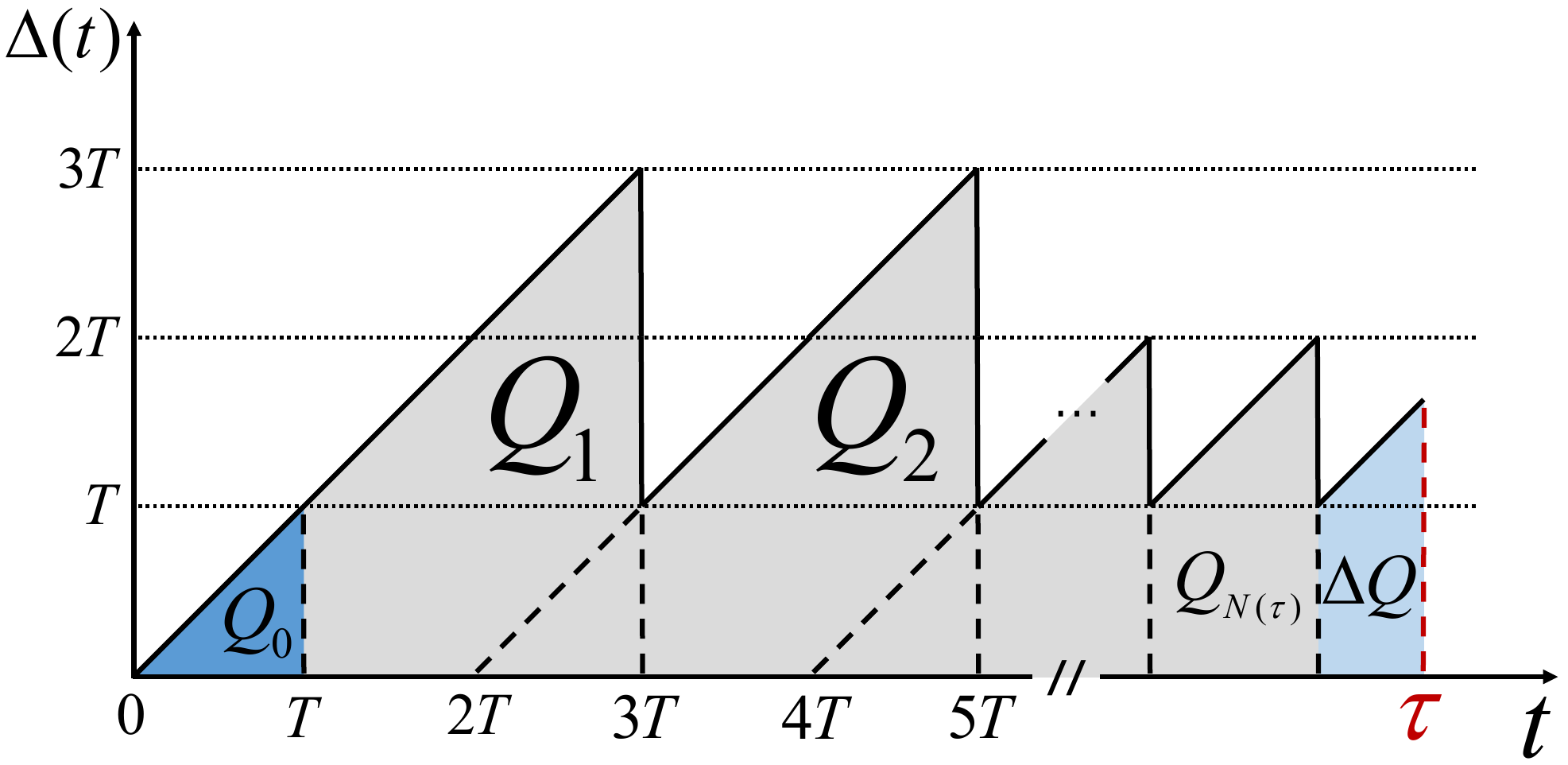}
\vspace{-5pt}
\caption{AoI evolution over time.}
\vspace{-10pt}
\label{fig3} 
\end{figure}
Fig.~\ref{fig3} illustrates the evolution of AoI over time, denoted as $\Delta(t)$. The AoI is defined as the elapsed time since the generation of the most recently successfully received information. When information reception fails, the AoI increases linearly with time, and upon a successful reception, the AoI is updated to $T$, corresponding to the time slot duration required for a complete transmission cycle. This standard linear AoI model effectively quantifies information freshness and ensures the tractability of mathematical analysis \cite{LineAoI}. Furthermore, we assume that the BS transmits only the pilots without the payloads in the first slot to acquire the initial CSIT. Since we focus on long-term information freshness, which is crucial for time-critical applications, this initial overhead is negligible. We adopt the AAoI as a metric to measure the long-term freshness performance, defined as
\begin{equation}
\label{19}
\bar{\Delta}=\lim_{\tau\to\infty}\frac{1}{\tau}\int_0^\tau\Delta(t) \, \mathrm{d}t.
\end{equation}

To facilitate the analysis of the AAoI, we denote the total number of successful receptions within time interval~$\tau$ by $N(\tau)$. The time integral of the AoI between the $(j-1)$-th and $j$-th successful receptions is denoted by $Q_j$, with $Q_0$ and $\Delta Q$ representing the area associated with the initial CSI estimation and the residual area at the end of the observation window, respectively. Based on the above definitions, Lemma~\ref{lemma6} provides the analytical expression for the AAoI under a constant reception failure probability $\epsilon$.
\begin{lemma}
\label{lemma6}
The AAoI can be expressed as
\begin{equation}
\label{20}
\bar{\Delta}=\frac{T}{2}+\frac{T}{1-\epsilon}.
\end{equation}
\end{lemma}
\begin{proof}
See Appendix~\ref{appendixD}.
\end{proof}

In the investigated scenario, the error probability varies dynamically due to random channel fading, making it difficult to derive closed-form AAoI expressions. To obtain a constant reception failure probability for tractable analysis, we approximate the instantaneous error probability using the average BLER. Theorems~\ref{theorem1} and \ref{theorem2} disclose the analytical AAoI for the common stream and the overall system, respectively.
\begin{theorem}
\label{theorem1}
The common stream AAoI at vehicle-$k$ is given by
\begin{equation}
\label{21}
\Delta_{c,k} = \frac{T}{2} + T \Bigl( \frac{\widetilde{\theta}_{2,k}}{\alpha_c}(2^{\frac{m_c}{n_c}}-1)+1 \Bigr)^{\widetilde{D}_{2,k}} \!\!
e^{\tfrac{2^{\frac{m_c}{n_c}}-1}{P_n\xi_k\alpha_c}}.
\end{equation}
\end{theorem}
\begin{theorem}
\label{theorem2}
The overall AAoI at vehicle-$k$ is given by
\begin{align}
\label{22}
\Delta_k =& \frac{T}{2} + T \Bigl( \frac{\widetilde{\theta}_{2,k}}{\alpha_c}(2^{\frac{m_c}{n_c}}-1)+1 \Bigr)^{\widetilde{D}_{2,k}} \nonumber \\
& \Bigl( \frac{\widetilde{\theta}_{3,k}}{\widetilde{D}_1\widetilde{\theta}_1\alpha_k}(2^{\frac{m_k}{n_k}}-1)+1 \Bigr)^{\widetilde{D}_{3,k}}\!\!
e^{\tfrac{2^{\frac{m_c}{n_c}}-1}{P_n\xi_k\alpha_c} + \tfrac{2^{\frac{m_k}{n_k}}-1}{P_n\xi_k\widetilde{D}_1\widetilde{\theta}_1\alpha_k}}.
\end{align}
\begin{proof}
Employ the average BLERs in Propositions~\ref{proposition1} and \ref{proposition2} to approximate the reception failure probability $\epsilon$ in Lemma~\ref{lemma6}. Consequently, substituting (\ref{16}) and (\ref{18}) into (\ref{20}) yields the analytical expressions for the common stream AAoI $\Delta_{c,k}$ and the overall AAoI $\Delta_k$.
\end{proof}
\end{theorem}

\section{Problem Formulation and Solution}
\label{section IV}
This section focuses on improving the AoI performance of the short-packet rate-splitting scheme. To this end, we jointly optimize the power allocation factors and the rate-splitting factors to minimize the average overall AAoI of vehicles. The optimization problem is formulated as follows:
\begin{subequations}
\label{P0}
\begin{align}
\mathcal{P}_0: \;
\min_{\alpha_c,\boldsymbol{\alpha},\boldsymbol{\psi}} \quad
&\frac{1}{K} \sum_{k\in\mathcal{K}} \Delta_k \\
\text{s.t.} \quad\;\;
& \alpha_c + \sum_{k\in\mathcal{K}}\alpha_k =1, \label{23b} \\
& 0 \le \alpha_c \le 1, \label{23c} \\
& 0 \le \alpha_k \le 1, \forall k\in\mathcal{K}, \label{23d} \\
& 0 \le \psi_k \le 1, \forall k\in\mathcal{K}, \label{23e} \\
& \Delta_{c,k} \le \lambda\Delta_k, \forall k\in\mathcal{K}, \label{23f}
\end{align}
\end{subequations}
where $\alpha_c$ and $\boldsymbol{\alpha} = [\alpha_1,\alpha_2,\dots,\alpha_K]$ represent the power allocation factors for the common and private streams, respectively, and $\boldsymbol{\psi} = [\psi_1,\psi_2,\dots,\psi_K]$ denotes the rate-splitting factors, indicating the proportion of each vehicle’s unicast message assigned to the common stream. Constraints (\ref{23b}), (\ref{23c}), and (\ref{23d}) ensure that the total transmit power is properly allocated among all streams. Constraint (\ref{23e}) guarantees the valid rate-splitting configuration. Constraint (\ref{23f}) regulates the rate splitting to ensure that the common stream AAoI remains lower than the overall AAoI, which can be used to meet stricter freshness requirements, where $\lambda$ is the constraint parameter.
\begin{remark}
The optimization aims to minimize the average overall AAoI. However, this objective may lead to excessive allocation of unicast messages to the common stream, thereby significantly deteriorating the common stream AAoI. Given that the overall AAoI gains are often marginal in such extreme cases, the QoS trade-off constraint (\ref{23f}) is introduced to prevent excessive rate splitting. With this constraint, the overall AAoI is only slightly compromised, while the common stream maintains satisfactory performance.
\end{remark}
The original problem \hyperref[P0]{$\mathcal{P}_0$} is highly non-convex and exhibits strong coupling among the variables $\alpha_c$, $\boldsymbol{\alpha}$, and $\boldsymbol{\psi}$, making a direct solution intractable. Furthermore, applying exhaustive search directly incurs prohibitive computational complexity. To reduce complexity while achieving effective optimization, we propose a multi-start two-step SCA algorithm. Specifically, the original problem \hyperref[P0]{$\mathcal{P}_0$} is decomposed into two subproblems. We initialize multiple rate-splitting starting points. For each starting point, the optimization proceeds in two steps: step~1 optimizes the power allocation factors via SCA, and step~2 optimizes the rate-splitting factors via SCA subject to constraint (\ref{23f}). Finally, the optimal solution is chosen from all starting points as the final result.
\begin{remark}
Given that the solution to non-convex problems depends on initialization, we adopt a multi-start strategy to address this. Extensive simulations indicate that initializing with larger rate-splitting starting points generally yields suboptimal performance. Moreover, the step~1 optimization is insensitive to minor variations in starting points. Therefore, only a few starting points are sufficient to cover the search space to obtain a desired result. Regarding the two-step optimization, we intentionally disregard constraint (\ref{23f}) in step~1. Since the constraint primarily regulates rate splitting, temporarily disregarding it simplifies the optimization and allows the power allocation to fully exploit its potential. Although the solution from step~1 may violate constraint (\ref{23f}), such cases are rare in typical settings. Even if a violation occurs, the step~2 optimization can effectively restore feasibility, which may cause the objective function to increase in the first iteration, but it still ensures efficient convergence thereafter.
\end{remark}

\subsection{Power Allocation Optimization}
In this subsection, rate splitting is temporarily ignored to focus on power allocation optimization. Accordingly, the rate-splitting constraint (\ref{23e}) and the QoS constraint (\ref{23f}) are omitted. To clearly distinguish the objective functions of the two-step optimization, we denote the overall AAoI as $\Delta_{k,\text{step1}}$ in this power allocation optimization, and the subproblem is formulated as:
\begin{subequations}
\label{P1}
\begin{align}
\mathcal{P}_1: \;
\min_{\alpha_c,\boldsymbol{\alpha}} \quad
&\frac{1}{K} \sum_{k\in\mathcal{K}} \Delta_{k,\text{step1}} \\
\text{s.t.} \quad\,
& \alpha_c + \sum_{k\in\mathcal{K}}\alpha_k =1, \\
& 0 \le \alpha_c \le 1, \\
& 0 \le \alpha_k \le 1, \forall k\in\mathcal{K}. 
\end{align}
\end{subequations}
We solve this non-convex problem using SCA. The notation $(\cdot)^{(t)}$ denotes the variable values obtained after the $t$-th SCA iteration. Based on these values, we subsequently construct the surrogate problem for \hyperref[P1]{$\mathcal{P}_1$} at the $(t+1)$-th iteration.

The strong coupling of $\alpha_c$ and $\boldsymbol{\alpha}$ across $\widetilde{D}_{2,k}$, $\widetilde{\theta}_{2,k}$, $\widetilde{D}_{3,k}$, and $\widetilde{\theta}_{3,k}$ renders \hyperref[P1]{$\mathcal{P}_1$} intractable. To address this, we introduce auxiliary variables to replace the coupled terms: $d_{2,k} = \widetilde{D}_{2,k}(\alpha_c,\boldsymbol{\alpha})$, $o_{2,k} = \widetilde{\theta}_{2,k}(\alpha_c,\boldsymbol{\alpha})$, $d_{3,k} = \widetilde{D}_{3,k}(\alpha_c,\boldsymbol{\alpha})$, $o_{3,k} = \widetilde{\theta}_{3,k}(\alpha_c,\boldsymbol{\alpha})$. The equality constraints introduced by the auxiliary variables are non-convex, and we relax them along the direction of optimization. Specifically, we convert these equalities into inequalities and further approximate the non-convex terms using first-order Taylor expansions. Accordingly, the relaxed constraints corresponding to $d_{2,k}$, $o_{2,k}$, $d_{3,k}$, and $o_{3,k}$ can be expressed as (\ref{25a}), (\ref{25b}), (\ref{25c}), and (\ref{25d}), respectively.
\begin{subequations}
\begin{align}
& \!\left[ (N_t-K+1)\rho^2\alpha_k + (1-\rho^2)(1-\alpha_c) \right]^2 + \frac{a_k^2+d_{2,k}^2}{2} \nonumber \\
& \quad \le \Bigl( a_k^{(t)}+d_{2,k}^{(t)} \Bigr) \Bigl(a_k+d_{2,k}-\frac{a_k^{(t)}+d_{2,k}^{(t)}}{2} \Bigr), \label{25a}\\
& (N_t-K+1)\rho^4\alpha_k^2 + (1-\rho^2)^2 \sum_{j\in\mathcal{K}}\alpha_j^2 - (1-\rho^2)o_{2,k} \nonumber \\
& \qquad + \frac{(N_t-K+1)\rho^2}{2}(\alpha_k^2+o_{2,k}^2) + \frac{1-\rho^2}{2}(\alpha_c+o_{2,k})^2 \nonumber \\
& \quad \le (N_t\!-\!K\!+\!1)\rho^2 \Bigl( \alpha_k^{(t)}\!+\!o_{2,k}^{(t)} \Bigr) \! \Bigl( \alpha_k\!+\!o_{2,k}\!-\!\frac{\alpha_k^{(t)}\!+\!o_{2,k}^{(t)}}{2} \Bigr) \nonumber \\
& \qquad + (1\!-\!\rho^2) \Bigl[ \alpha_c^{(t)}\alpha_c \!+\! o_{2,k}^{(t)}o_{2,k} \!-\! \frac{(\alpha_c^{(t)})^2\!+\!(o_{2,k}^{(t)})^2}{2} \Bigr],\! \label{25b}\\
& (1-\alpha_c-\alpha_k)^2 + \frac{b_k^2+d_{3,k}^2}{2} \nonumber \\
& \quad \le \Bigl( b_k^{(t)}+d_{3,k}^{(t)} \Bigr) \Bigl( b_k+d_{3,k}-\frac{b_k^{(t)}+d_{3,k}^{(t)}}{2} \Bigr), \label{25c} \\
& (1-\rho^2)\!\!\!\sum_{j\in\mathcal{K}\setminus\{k\}}\!\!\!\alpha_j^2 + \frac{(\alpha_k+o_{3,k})^2}{2} + \frac{(\alpha_c+o_{3,k})^2}{2} - o_{3,k} \nonumber \\
& \quad \le \alpha_k^{(t)} \! \Bigl( \alpha_k\!-\!\frac{\alpha_k^{(t)}}{2} \Bigr) \!+\! \alpha_c^{(t)}\!\Bigl( \alpha_c\!-\!\frac{\alpha_c^{(t)}}{2} \Bigr) \!+\! o_{3,k}^{(t)}\Bigl( 2o_{3,k}\!-\!o_{3,k}^{(t)} \Bigr), \label{25d}
\end{align}
\end{subequations}
where $a_k=(N_t-K+1)\rho^4\alpha_k^2+(1-\rho^2)^2\sum_{j\in\mathcal{K}}\alpha_j^2$ and $b_k=\sum_{j\in\mathcal{K}\setminus\{k\}}\alpha_j^2$ are auxiliary variables. The corresponding relaxed constraints can similarly be expressed as (\ref{26a}) and (\ref{26b}), respectively.
\begin{subequations}
\begin{align}
a_k \le&\; (N_t-K+1)\rho^4 \Bigl[ 2\alpha_k^{(t)}\alpha_k-(\alpha_k^{(t)})^2 \Bigr] \nonumber \\
& + (1-\rho^2)^2\sum_{j\in\mathcal{K}} \Bigl[ 2\alpha_j^{(t)}\alpha_j-(\alpha_j^{(t)})^2 \Bigr], \label{26a} \\
b_k \le& \!\!\!\sum_{j\in\mathcal{K}\setminus\{k\}}\! \Bigl[ 2\alpha_j^{(t)}\alpha_j-(\alpha_j^{(t)})^2 \Bigr]. \label{26b}
\end{align}
\end{subequations}

To facilitate subsequent optimization, we denote the overall AAoI as $\Delta_{k,\text{step1}} = T/2 + Tg_k$, the non-optimized terms as $A_c=2^{m_c/n_c}-1$ and $A_k= (2^{m_k/n_k}-1) / (\widetilde{D}_1\widetilde{\theta}_1)$. Based on (\ref{22}) and by substituting $A_c$ and $A_k$, $g_k$ can be expressed as
\begin{equation}
\label{27}
g_k = \Bigl( A_c\frac{o_{2,k}}{\alpha_c}\!+\!1 \Bigr)^{\!d_{2,k}} \! \Bigl( A_k\frac{ o_{3,k}}{\alpha_k}\!+\!1 \Bigr)^{\!d_{3,k}} \! e^{\tfrac{1}{P_n\xi_k}\!\Bigl(\! \tfrac{A_c}{\alpha_c}+\tfrac{A_k}{\alpha_k} \!\Bigr)}.
\end{equation}
We first apply a logarithmic transformation to (\ref{27}), yielding
\begin{align}
\label{28}
\ln g_k =&\; d_{2,k} \ln \Bigl( A_c\frac{o_{2,k}}{\alpha_c}+1 \Bigr) + d_{3,k} \ln \Bigl( A_k\frac{o_{3,k}}{\alpha_k}+1 \Bigr) \nonumber\\
&+ \frac{A_c}{P_n\xi_k\alpha_c} + \frac{A_k}{P_n\xi_k\alpha_k}.
\end{align}
According to the concavity property of logarithmic functions~\cite{logarithmic}, we apply first-order Taylor expansions to obtain upper bounds for the logarithmic terms in (\ref{28}). To handle the resulting product terms, we apply Young’s inequality $xy \le x^p/p+y^q/q$~\cite{inequality}, choosing $p=q=2$ and substituting the variables as $x\sqrt{y^{(t)}/x^{(t)}} \rightarrow x$, $y\sqrt{x^{(t)}/y^{(t)}} \rightarrow y$, which yields $xy \le \frac{1}{2}(x^2y^{(t)}/x^{(t)}+y^2x^{(t)}/y^{(t)})$. By further applying
this inequality, the convex upper bound of $\ln g_k$ can be expressed as
\begin{align}
\ln \!g_k \!\le&\; \widetilde{g}_k^{(t)} \nonumber \\
=& \ln(A_c c_k^{(t)}\!+\!1)d_{2,k} \!+\! \frac{A_c}{2A_c c_k^{(t)}\!+\!2} \Bigl( c_k^2\frac{d_{2,k}^{(t)}}{c_k^{(t)}} \!+\! d_{2,k}^2\frac{c_k^{(t)}}{d_{2,k}^{(t)}} \Bigr) \nonumber \\
& +\! \ln(A_k e_k^{(t)}\!+\!1)d_{3,k} \!+\! \frac{A_k}{2 A_k e_k^{(t)}\!+\!2} \Bigl( e_k^2\frac{d_{3,k}^{(t)}}{e_k^{(t)}} \!+\! d_{3,k}^2\frac{e_k^{(t)}}{d_{3,k}^{(t)}} \Bigr) \nonumber \\
& +\! \frac{A_c}{P_n\xi_k\alpha_c} \!+\! \frac{A_k}{P_n\xi_k\alpha_k} \!-\! \frac{A_c c_k^{(t)}}{A_c c_k^{(t)}\!+\!1}d_{2,k} \!-\! \frac{A_k e_k^{(t)}}{A_k e_k^{(t)}\!+\!1}d_{3,k},
\end{align}
where $c_k = o_{2,k}/\alpha_c$ and $e_k = o_{3,k}/\alpha_k$ are auxiliary variables, and the corresponding relaxed constraints can be expressed as
\begin{subequations}
\begin{align}
o_{2,k} \!+\! \frac{\alpha_c^2\!+\!c_k^2}{2} &\le \Bigl( \alpha_c^{(t)}\!+\!c_k^{(t)} \Bigr)\! \Bigl( \alpha_c\!+\!c_k\!-\!\frac{\alpha_c^{(t)}\!+\!c_k^{(t)}}{2} \Bigr), \label{30a} \\
o_{3,k} \!+\! \frac{\alpha_k^2\!+\!e_k^2}{2} &\le \Bigl( \alpha_k^{(t)}\!+\!e_k^{(t)} \Bigr)\! \Bigl( \alpha_k\!+\!e_k\!-\!\frac{\alpha_k^{(t)}\!+\!e_k^{(t)}}{2} \Bigr). \label{30b}
\end{align}
\end{subequations}
Therefore, the surrogate function of $\Delta_{k,\text{step1}}$ is given by
\begin{equation}
\widetilde{\Delta}_{k,\text{step1}}^{(t)}=\frac{T}{2}+Te^{\widetilde{g}_k^{(t)}}.
\end{equation}

Let $\Omega = \{ d_{2,k},o_{2,k},d_{3,k},o_{3,k},a_k,b_k,c_k,e_k \}_{k \in \mathcal{K}}$ denote the set of all auxiliary variables. Finally, the subproblem \hyperref[P1]{$\mathcal{P}_1$} is transformed into \hyperref[P2]{$\mathcal{P}_2$} at the $(t+1)$-th iteration, given by
\begin{subequations}
\label{P2}
\begin{align}
\mathcal{P}_2: 
\min_{\alpha_c,\boldsymbol{\alpha},\Omega} \quad 
& \frac{1}{K} \sum_{k\in\mathcal{K}} \widetilde{\Delta}_{k,\text{step1}}^{(t)} \\
\text{s.t.} \;\;\quad
& \alpha_c + \sum_{k\in\mathcal{K}}\alpha_k =1, \\
& 0 \le \alpha_c \le 1,\\
& 0 \le \alpha_k \le 1, \forall k\in\mathcal{K},\\
& (\ref{25a})\text{--}(\ref{25d}), \forall k\in\mathcal{K}, \\
& (\ref{26a}), (\ref{26b}), \forall k\in\mathcal{K}, \\
& (\ref{30a}), (\ref{30b}), \forall k\in\mathcal{K}.
\end{align}
\end{subequations}
The surrogate problem \hyperref[P2]{$\mathcal{P}_2$} is convex and can be directly solved using the CVX toolbox. By iteratively solving it, \hyperref[P1]{$\mathcal{P}_1$} is solved to determine the optimal power allocation.

\subsection{Rate-Splitting Optimization}
After optimizing power allocation, we proceed to optimize rate splitting. Assigning a portion of a unicast message to the common stream increases its transmission load while reducing that of the corresponding private stream. This indicates that an appropriate rate-splitting strategy can further improve the overall AAoI. We denote the overall AAoI as $\Delta_{k,\text{step2}}$ for clarity and impose the QoS constraint (\ref{23f}) to regulate the rate splitting. The subproblem is then formulated as:
\begin{subequations}
\label{P3}
\begin{align}
\mathcal{P}_3:\;
\min_{\boldsymbol{\psi}} \quad 
&\frac{1}{K} \sum_{k\in\mathcal{K}} \Delta_{k,\text{step2}} \\
\text{s.t.} \quad\,
& 0 \le \psi_k \le 1, \forall k\in\mathcal{K}, \\
& \Delta_{c,k} \le \lambda\Delta_{k,\text{step2}}, \forall k\in\mathcal{K}. \label{33c}
\end{align}
\end{subequations}
Similarly, SCA is employed to solve the subproblem and $(\cdot)^{(t)}$ is used to denote the variable values obtained after the $t$-th iteration. Subsequently, we construct the surrogate problem for \hyperref[P3]{$\mathcal{P}_3$} at the $(t+1)$-th iteration.

We assume that each user has the same communication requirement of transmitting $m_0$ information bits, which consist of the multicast message and unicast message, and use $k_0$ to represent the proportion of the multicast message. Therefore, the number of information bits of the common and the $k$-th private streams are given by $m_c = k_0m_0+(1-k_0)m_0\sum_{j\in\mathcal{K}}\psi_j$ and $m_k = (1-k_0)(1-\psi_k)m_0$, respectively. To simplify the objective function, we introduce auxiliary variables $\beta_c = 2^{m_c/n_c}-1$ and $\beta_k = 2^{m_k/n_k}-1$, and the corresponding constraints can be relaxed as
\begin{subequations}
\begin{align}
k_0 + (1-k_0)\sum_{j\in\mathcal{K}}\psi_j - \frac{n_c}{m_0}\log_2(\beta_c+1) &\le 0, \label{34a} \\
(1-k_0)(1-\psi_k) - \frac{n_k}{m_0}\log_2(\beta_k+1) &\le 0. \label{34b} 
\end{align}
\end{subequations}
Then, we denote the overall AAoI as $\Delta_{k,\text{step2}} = T/2 + Th_k$. Based on (\ref{22}) and by substituting $\beta_c$ and $\beta_k$, $h_k$ can be expressed as
\begin{equation}
\label{35}
h_k \!=\! \Bigl( \frac{\widetilde{\theta}_{2,k}}{\alpha_c}\beta_c + 1 \Bigr)^{\!\widetilde{D}_{2,k}}\! \Bigl( \frac{\widetilde{\theta}_{3,k}}{\widetilde{D}_1\widetilde{\theta}_1\alpha_k}\beta_k +1 \Bigr)^{\!\widetilde{D}_{3,k}}\!\! e^{\tfrac{1}{P_n\xi_k}\! \Bigl(\! \tfrac{\beta_c}{\alpha_c} + \tfrac{\beta_k}{\widetilde{D}_1\widetilde{\theta}_1\alpha_k} \!\Bigr)}.
\end{equation}
Applying a logarithmic transformation to (\ref{35}) yields
\begin{align}
\label{36}
\ln h_k =&\; \widetilde{D}_{2,k}\ln \Bigl( \frac{\widetilde{\theta}_{2,k}}{\alpha_c}\beta_c+1 \Bigr) + \widetilde{D}_{3,k}\ln \Bigl( \frac{\widetilde{\theta}_{3,k}}{\widetilde{D}_1\widetilde{\theta}_1\alpha_k}\beta_k+1 \Bigr) \nonumber\\
&+ \frac{\beta_c}{P_n\xi_k\alpha_c}+\frac{\beta_k}{P_n\xi_k\widetilde{D}_1\widetilde{\theta}_1\alpha_k}.
\end{align}
Since the power allocation has been determined, (\ref{36}) only depends on $\beta_c$ and $\beta_k$. Similarly, we apply first-order Taylor expansions to obtain upper bounds for the logarithmic terms in (\ref{36}), which yields the convex upper bound of $\ln h_k$:
\begin{align}
\ln h_k \le&\; \widetilde{h}_k^{(t)} \nonumber\\
=&\; \widetilde{D}_{2,k} \Bigl[ \ln \Bigl( \frac{\widetilde{\theta}_{2,k}\beta_c^{(t)}}{\alpha_c}+1 \Bigr) + \frac{\widetilde{\theta}_{2,k}(\beta_c-\beta_c^{(t)})}{\widetilde{\theta}_{2,k}\beta_c^{(t)}+\alpha_c} \Bigr] \nonumber\\
& + \widetilde{D}_{3,k} \Bigl[ \ln \Bigl( \frac{\widetilde{\theta}_{3,k}\beta_k^{(t)}}{\widetilde{D}_1\widetilde{\theta}_1\alpha_k}+1 \Bigr) + \frac{\widetilde{\theta}_{3,k}(\beta_k-\beta_k^{(t)})}{\widetilde{\theta}_{3,k}\beta_k^{(t)}+\widetilde{D}_1\widetilde{\theta}_1\alpha_k} \Bigr] \nonumber\\ 
& + \frac{\beta_c}{P_n\xi_k\alpha_c} + \frac{\beta_k}{P_n\xi_k\widetilde{D}_1\widetilde{\theta}_1\alpha_k}.
\end{align}
Therefore, the surrogate function of $\Delta_{k,\text{step2}}$ is given by
\begin{equation}
\widetilde{\Delta}_{k,\text{step2}}^{(t)} = \frac{T}{2} + Te^{\widetilde{h}_k^{(t)}}.
\end{equation}

In addition, we denote the common stream AAoI as $\Delta_{c,k} = T/2 + Th_{c,k}$, so constraint (\ref{33c}) can be rewritten as
\begin{equation}
\label{39}
h_{c,k} \le \lambda h_k + \frac{\lambda-1}{2},
\end{equation}
where $h_{c,k}$ can be expressed as
\begin{equation}
\label{40}
h_{c,k} = \Bigl( \frac{\widetilde{\theta}_{2,k}}{\alpha_c}\beta_c+1 \Bigr)^{\widetilde{D}_{2,k}} e^{\tfrac{\beta_c}{P_n\xi_k\alpha_c}}.
\end{equation}
Constraint (\ref{39}) is non-convex, where the parameter term $(\lambda-1)/2$ makes it difficult to transform into a convex form. Thus, we neglect this term and approximate (\ref{39}) as $h_{c,k} \le \lambda h_k$. By substituting (\ref{35}) and (\ref{40}), and then taking a logarithmic transformation, the relaxed constraint can be expressed as
\begin{equation}
\label{41}
-\widetilde{D}_{3,k} \ln \Bigl( \frac{\widetilde{\theta}_{3,k}}{\widetilde{D}_1\widetilde{\theta}_1\alpha_k}\beta_k+1 \Bigr)
 - \frac{\beta_k}{P_n\xi_k\widetilde{D}_1\widetilde{\theta}_1\alpha_k} \le \ln \lambda.
\end{equation}
\begin{remark}
In fact, neglecting the parameter term enlarges the feasible set. However, this relaxation has a negligible impact on the validity of the constraint. The relaxed constraint (\ref{41}) still effectively regulates the rate-splitting factors $\boldsymbol{\psi}$ by restricting the relationship between the common stream AAoI and the overall AAoI.
\end{remark}
Finally, the subproblem \hyperref[P3]{$\mathcal{P}_3$} is transformed into \hyperref[P4]{$\mathcal{P}_4$} at the $(t+1)$-th iteration, given by
\begin{subequations}
\label{P4}
\begin{align}
\mathcal{P}_4: \;
\min_{ \boldsymbol{\psi},\beta_c,\{ \beta_k\! \}_{\! k \mkern-1.5mu\in\mkern-1.5mu \mathcal{K}} } \quad
& \frac{1}{K} \sum_{k\in\mathcal{K}} \widetilde{\Delta}_{k,\text{step2}}^{(t)} \\
\text{s.t.} \qquad\;\,\,
& 0 \le \psi_k \le 1, \forall k\in\mathcal{K}, \\
& (\ref{34a}), \\
& (\ref{34b}), \forall k\in\mathcal{K}, \\
& (\ref{41}), \forall k\in\mathcal{K}.
\end{align}
\end{subequations}
The surrogate problem \hyperref[P4]{$\mathcal{P}_4$} is convex and can be directly solved with the CVX toolbox. By iteratively solving it, we can solve \hyperref[P3]{$\mathcal{P}_3$} and obtain the optimal rate splitting. 

In summary, the overall algorithm for solving \hyperref[P0]{$\mathcal{P}_0$} is outlined in Algorithm \ref{algorithm1}.
\begin{algorithm}[t]
\caption{The Multi-Start Two-Step SCA Algorithm}
\label{algorithm1}
\begin{algorithmic}[1]
\STATE Initialize multiple starting points for $\boldsymbol{\psi}$.
\FOR{each starting point $\boldsymbol{\psi}^{(i)}$}
    \STATE Optimize $\alpha_c$ and $\boldsymbol{\alpha}$ by iteratively solving \hyperref[P2]{$\mathcal{P}_2$} with $\boldsymbol{\psi}^{(i)}$ fixed.
    \STATE Optimize $\boldsymbol{\psi}$ by iteratively solving \hyperref[P4]{$\mathcal{P}_4$} with $\alpha_c^{*(i)}$ and $\boldsymbol{\alpha}^{*(i)}$ fixed.
\ENDFOR
\STATE Select the solution with the minimum objective function value among all starting points as the final result for \hyperref[P0]{$\mathcal{P}_0$}.
\STATE \textbf{Return:} $\alpha_c^*$, $\boldsymbol{\alpha}^*$, $\boldsymbol{\psi}^*$.
\end{algorithmic}
\end{algorithm}

\subsection{Complexity Analysis}
The algorithm decomposes the original problem into two subproblems and sequentially optimizes power allocation and rate splitting via SCA. To mitigate potential bias from different initializations, we adopt the multi-start strategy that selects a few starting points and finally chooses the best solution. Since this strategy only affects the search process and does not alter the complexity of each optimization, it can be neglected in the complexity analysis. Consequently, the total complexity only depends on the two-step SCA procedure, where each iteration applies the interior-point method to solve. Specifically, \hyperref[P2]{$\mathcal{P}_2$} involves $9K+1$ variables and $9K+2$ constraints, yielding a complexity of $\mathcal{O}\left( N_{1}(9K+1)^{3}\sqrt{9K+2} \right)$, where $N_{1}$ denotes the number of SCA iterations. Similarly, \hyperref[P4]{$\mathcal{P}_4$} involves $2K+1$ variables and $3K+1$ constraints, resulting in a complexity of $\mathcal{O}\left( N_{2}(2K+1)^{3}\sqrt{3K+1} \right)$, where $N_{2}$ is the number of SCA iterations. Therefore, the total computational complexity is $\mathcal{O}\left( N_{1}(9K\!+\!1)^{3}\sqrt{9K\!+\!2} + N_{2}(2K\!+\!1)^{3}\sqrt{3K\!+\!1} \right)$, which can be further simplified as $\mathcal{O}\left( (N_{1}+N_{2})K^{3.5} \right)$.

\section{Simulation Results}
\label{section V}
In this section, simulations are carried out to evaluate the validity of the theoretical analysis and the AAoI performance of the proposed RSMA scheme. Unless otherwise specified, the default parameter configurations are as follows: the BS is equipped with five transmit antennas ($N_t=5$) and serves four single-antenna vehicles ($K=4$). The total transmit power is $P_T=35$~dBm~\cite{3gpp.37.885}. The carrier frequency is $f_c=5.9$~GHz~\cite{5.9GHz}, with a bandwidth of $B=10$~MHz and the noise power density of $\sigma^2=-174$~dBm/Hz. The vehicle velocity is $v=200$~km/h, and the transmission interval is $T=0.24$~ms. The number of information bits is $m_0=200$~bits, with the multicast message proportion of $k_0=0.3$. The blocklength is $n_c=n_k=400$ for all $k \in \mathcal{K}$~\cite{blocklength}. The distances between the BS and vehicles are uniformly distributed between 200 and 350 meters. The QoS parameter is set as $\lambda=0.8$. All simulation results are obtained via a Monte Carlo method based on $10^5$ channel realizations.

\begin{figure}[t] 
\centering 
\includegraphics[width=0.85\linewidth]{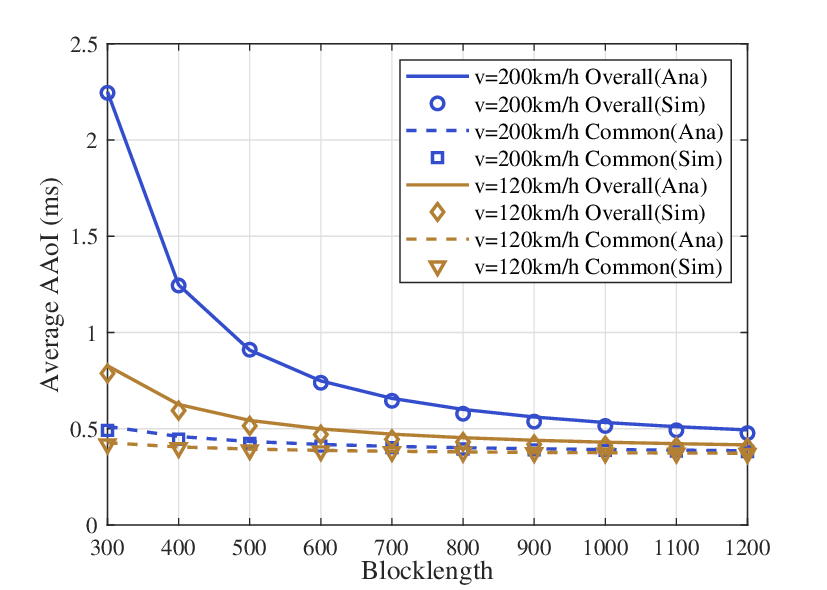}
\vspace{-10pt}
\caption{Comparison of the analytical and simulated average AAoI.}
\vspace{-15pt}
\label{fig4}
\end{figure}
We first verify the validity of the derived analytical expressions. As shown in Fig.~\ref{fig4}, the analytical results for the average AAoI of vehicles obtained from (\ref{21}) and (\ref{22}) closely match the simulation results at the two considered vehicle velocities. This close agreement confirms the accuracy of the closed-form expressions, which are employed to optimize the AAoI performance of the proposed RSMA scheme.

\begin{figure}[t] 
\centering 
\includegraphics[width=0.85\linewidth]{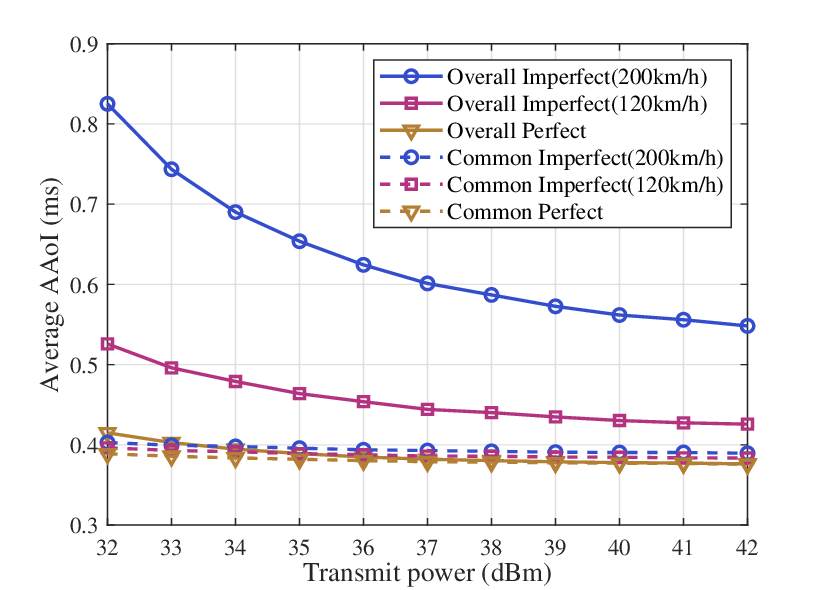}
\vspace{-10pt}
\caption{Comparison of the average AAoI under perfect and imperfect CSIT.}
\vspace{-10pt}
\label{fig5} 
\end{figure}
Fig.~\ref{fig5} illustrates the average AAoI under both perfect and imperfect CSIT. In the latter case, two vehicle velocities are considered, where the higher velocity corresponds to less accurate CSIT. It can be observed that the common stream AAoI is consistently much lower than the overall AAoI. Thus, time-sensitive information can be assigned to the common stream to satisfy its stricter AAoI requirements. Furthermore, as CSIT accuracy deteriorates, the common stream AAoI remains nearly unchanged, while the overall AAoI increases significantly. This indicates that the common stream is insensitive to CSIT accuracy, and the overall performance is highly dependent on it. Therefore, under imperfect CSIT, an appropriate rate splitting between the common and private streams can improve the overall performance.

\begin{figure}[t]
    \centering
    \includegraphics[width=0.85\linewidth]{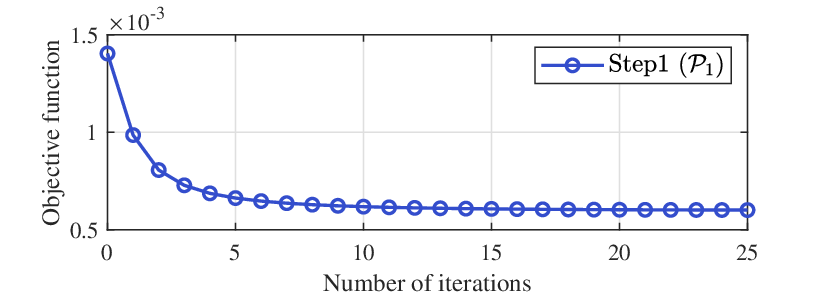}
    \\
    \small (a) Convergence of subproblem $\mathcal{P}_1$.
    \\
    \includegraphics[width=0.85\linewidth]{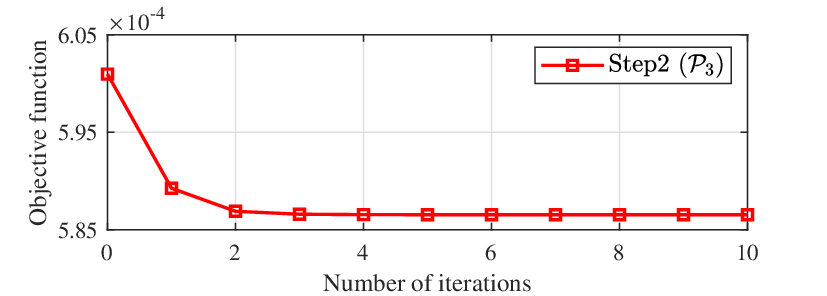} 
    \\
    \small (b) Convergence of subproblem $\mathcal{P}_3$.
    \vspace{-5pt}
    \caption{Convergence of the proposed algorithm.}
    \vspace{-15pt}
    \label{fig6}
\end{figure}
We optimize the AAoI performance of RSMA using the derived closed-form expressions and the proposed algorithm. Fig.~\ref{fig6} illustrates the convergence trend of the two-step SCA iterations in the proposed algorithm. By iteratively solving \hyperref[P2]{$\mathcal{P}_2$} and \hyperref[P4]{$\mathcal{P}_4$}, subproblem \hyperref[P1]{$\mathcal{P}_1$} and \hyperref[P3]{$\mathcal{P}_3$} can be efficiently resolved, respectively. It can be observed that the objective function gradually converges as the number of iterations increases, demonstrating the desirable convergence property of the proposed algorithm. It is noteworthy that subproblem \hyperref[P1]{$\mathcal{P}_1$} requires more iterations to converge, since the introduction of more auxiliary variables leads to slower updates in each iteration.

In the following, we compare the proposed RSMA scheme with NOMA and SDMA to evaluate the AAoI performance. Specifically, the comparison schemes include:
\begin{itemize}
\item RSMA, with power allocation and rate-splitting optimization using the proposed algorithm, where the common and overall performance are termed RSMA (Common) and RSMA (Overall), respectively.
\item SDMA, based on ZF precoding with exhaustive power allocation using a step size of 0.01~\cite{SDMA}.
\item NOMA, based on inter-group spatial division multiplexing, where four users are paired into two groups by proximity, and one layer of SIC is employed within each group~\cite{NOMA}. Exhaustive inter-group and intra-group power allocations are performed using a step size of 0.01.
\end{itemize}

To better demonstrate the performance of RSMA, the principal eigenvector of the CSIT covariance matrix is adopted as the precoder for the common stream.

\begin{figure}[t] 
\centering 
\includegraphics[width=0.85\linewidth]{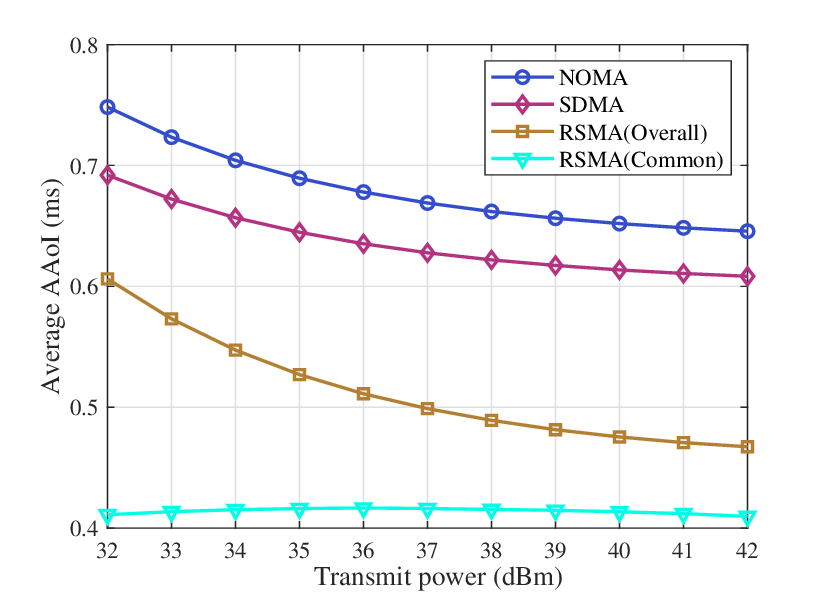}
\vspace{-10pt}
\caption{Average AAoI versus transmit power.}
\vspace{-15pt}
\label{fig7} 
\end{figure}
Fig.~\ref{fig7} shows the average AAoI of vehicles with respect to transmit power. It can be observed that the RSMA scheme consistently outperforms both SDMA and NOMA within the considered transmit power range. This advantage stems from effective resource management through rate splitting, which enables RSMA to better cope with imperfect CSIT. Furthermore, as the transmit power increases, the average AAoI of RSMA (Overall), SDMA, and NOMA all decrease, with the performance advantage of RSMA (Overall) becoming increasingly pronounced, whereas the average AAoI of RSMA (Common) remains virtually unchanged. This observation can be explained as follows. On one hand, higher transmit power provides RSMA greater flexibility in resource allocation, yielding the greatest performance gains for RSMA (Overall). On the other hand, increased transmit power enhances the transmission capacity of the common stream, prompting the algorithm to assign more messages to it, which offsets the potential gains from the increased power and keeps RSMA (Common) nearly unchanged. Nevertheless, this sacrifice of the common stream still ensures that RSMA (Common) achieves optimal performance among all schemes and further improves the performance of RSMA (Overall).

\begin{figure}[t] 
\centering 
\includegraphics[width=0.85\linewidth]{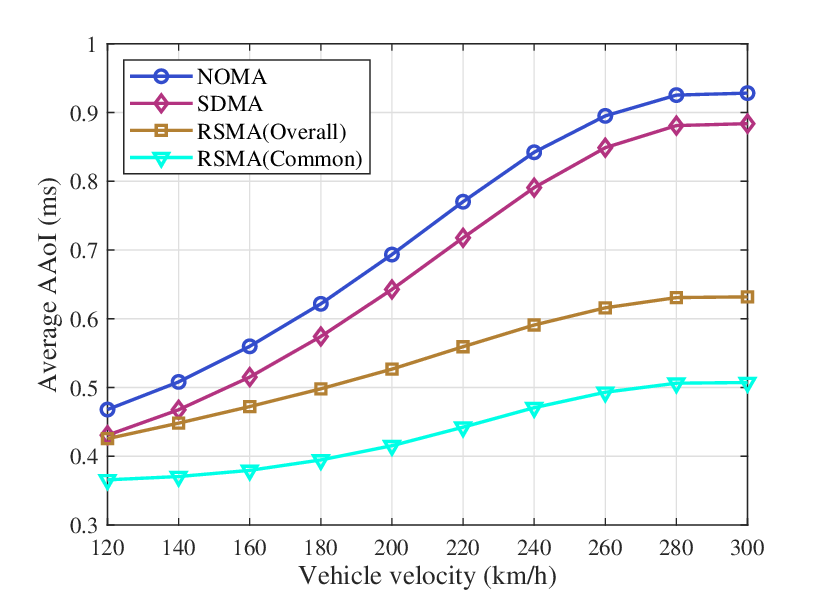}
\vspace{-10pt}
\caption{Average AAoI versus vehicle velocity.}
\vspace{-15pt}
\label{fig8} 
\end{figure}
Fig.~\ref{fig8} depicts the variation of the average AAoI with respect to vehicle velocity. As shown, the RSMA scheme consistently achieves the lowest average AAoI, demonstrating its superior robustness against vehicle mobility. It can be observed that the average AAoI across all schemes increases significantly as velocity increases, since higher velocity reduces CSIT accuracy, leading to more imperfect precoding. The performance deterioration of RSMA is much lower than that of SDMA and NOMA, resulting in an increasingly pronounced performance advantage, indicating that rate splitting effectively mitigates the adverse effects of imperfect CSIT. Additionally, RSMA (Overall) and SDMA exhibit comparable performance at lower velocities. This is because precoding is sufficiently accurate at lower velocities, making the private streams of RSMA highly reliable and reducing reliance on the common stream, thereby causing RSMA to gradually degenerate into SDMA.

\begin{figure}[t] 
\centering 
\includegraphics[width=0.85\linewidth]{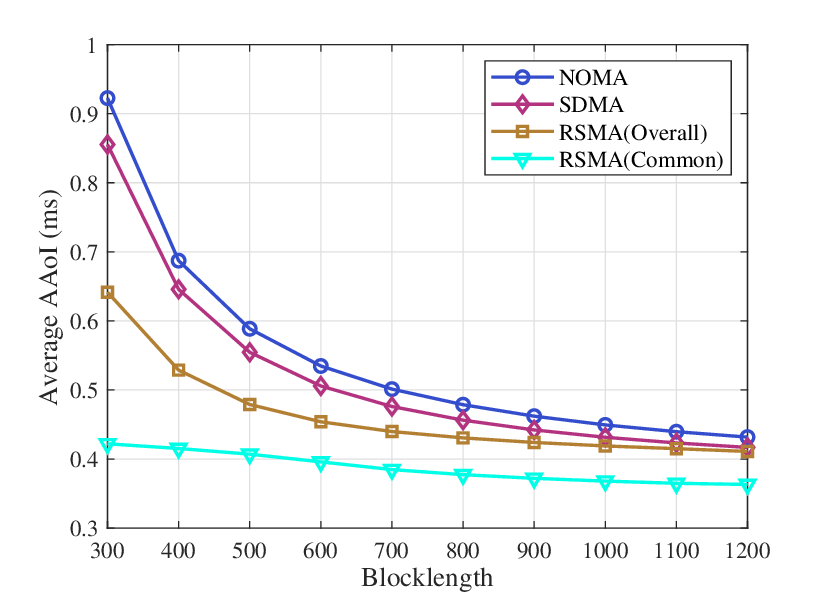}
\vspace{-10pt}
\caption{Average AAoI versus blocklength.}
\vspace{-15pt}
\label{fig9} 
\end{figure}
Fig.~\ref{fig9} displays the performance versus blocklength under different schemes.\footnote{We exclude extremely short blocklengths as imperfect CSIT induces excessive errors in the absence of retransmissions. Given a 10~MHz bandwidth, the latency of the considered blocklength range still meets URLLC requirements.} As the blocklength increases, the average AAoI of all schemes decreases due to reduced decoding errors, with the RSMA scheme consistently achieving the best performance, demonstrating its superior capability to achieve higher information freshness under the same blocklength. Moreover, the performance advantage of RSMA becomes more significant at shorter blocklengths, where the effect of FBL is more critical. As the blocklength approaches infinity, the average AAoI of all schemes gradually converges to the lower bound corresponding to the ideal case with zero BLER.

\begin{figure}[t] 
\centering 
\includegraphics[width=0.85\linewidth]{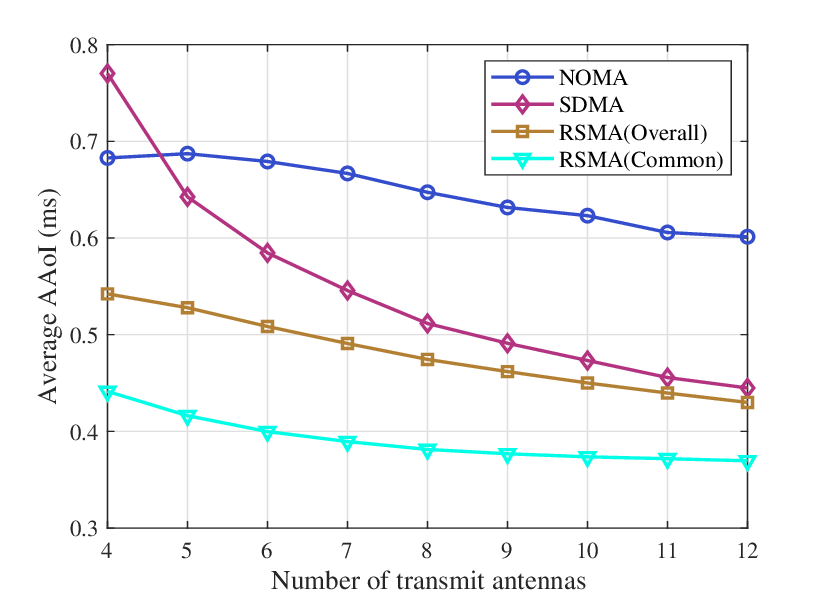}
\vspace{-10pt}
\caption{Average AAoI versus the number of transmit antennas.}
\vspace{-15pt}
\label{fig10} 
\end{figure}
Fig.~\ref{fig10} demonstrates the performance versus the number of transmit antennas at the BS under different schemes. As the number of antennas increases, the spatial degrees of freedom of the system enhance, leading to a reduction in the average AAoI for all schemes, with the RSMA scheme consistently achieving optimal performance, indicating its relatively low dependence on antenna resources. When the number of antennas equals the number of vehicles, the limitations of ZF precoding result in the worst performance for SDMA, whereas RSMA still maintains strong performance due to its superior robustness. Furthermore, an increased number of antennas improves the effectiveness of ZF precoding, thereby significantly enhancing the performance of SDMA, while still being slightly inferior to RSMA (Overall) even at its best.

\begin{figure}[t] 
\centering 
\includegraphics[width=0.85\linewidth]{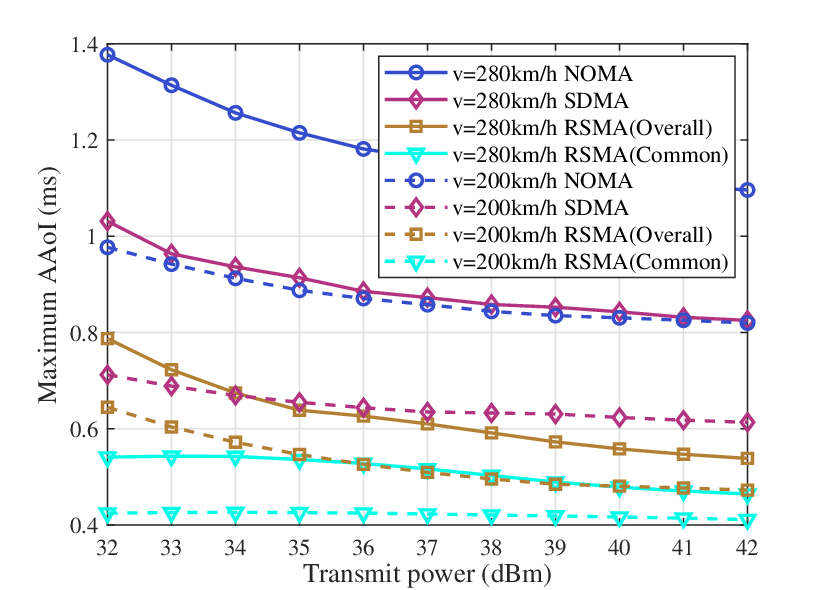}
\vspace{-10pt}
\caption{Maximum AAoI versus transmit power.}
\vspace{-15pt}
\label{fig11} 
\end{figure}
\begin{figure}[t] 
\centering 
\includegraphics[width=0.85\linewidth]{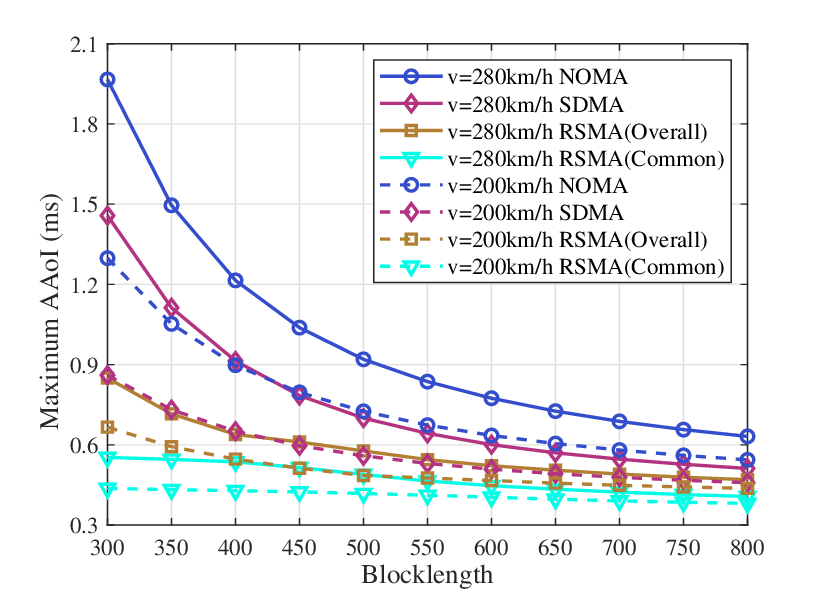}
\vspace{-10pt}
\caption{Maximum AAoI versus blocklength.}
\vspace{-15pt}
\label{fig12} 
\end{figure}
Since the performance of each user is considered equally important, we next evaluate the performance of the worst-case user. Fig.~\ref{fig11} and Fig.~\ref{fig12} illustrate the maximum AAoI among vehicles with respect to transmit power and blocklength, respectively, considering two vehicle velocities. It is evident that the RSMA scheme consistently achieves the lowest maximum AAoI, indicating that the weakest user in RSMA outperforms that in both SDMA and NOMA. Although the proposed algorithm aims to minimize the average AAoI, it does not excessively degrade the performance of individual users, achieving an effective balance between overall and worst-case performance. This advantage arises not only from the appropriate power allocation between the common and private streams but more importantly from the effective rate splitting, which ensures fairness among users. In contrast, SDMA and NOMA perform only power allocation and lack the flexible resource scheduling provided by rate splitting, limiting their capability to mitigate performance differences caused by channel variations. Taken together, Figs.~\ref{fig7}--\ref{fig12} demonstrate the performance advantages of the proposed short-packet RSMA scheme, which consistently achieves the lowest average AAoI, maintains fairness among users, and delivers significantly lower AAoI via the common stream.

\begin{figure}[t]
\centering 
\includegraphics[width=0.85\linewidth]{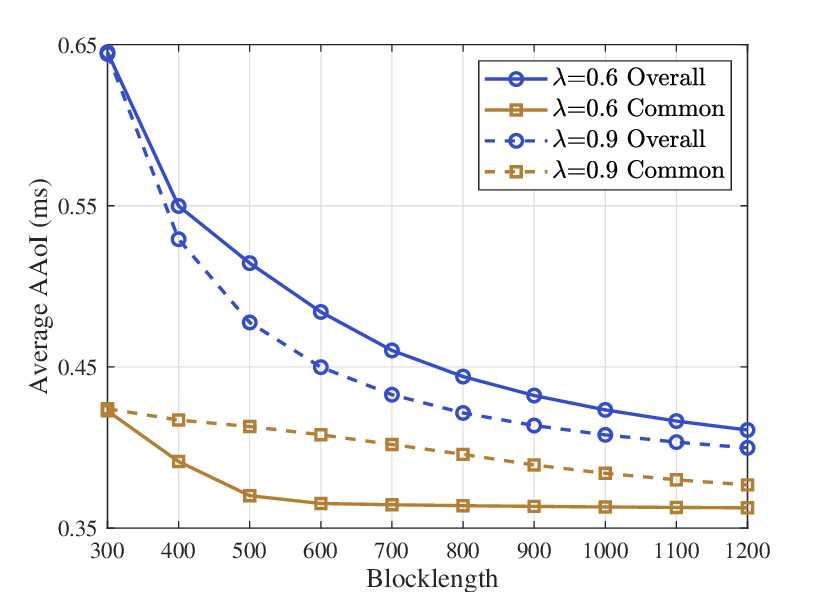}
\vspace{-10pt}
\caption{The trade-off effect of $\lambda$ on the average AAoI.}
\vspace{-15pt}
\label{fig13} 
\end{figure}
Fig.~\ref{fig13} shows the variation of the average AAoI of RSMA versus blocklength under two values of $\lambda$. Particularly, $\lambda$ is a constraint parameter as shown in (\ref{23f}), in the hope of regulating the rate splitting to meet stricter and flexible freshness requirements between common and private streams. It can be observed that a larger $\lambda$ yields lower overall AAoI but higher common stream AAoI. This is because a larger $\lambda$ relaxes the AAoI requirement of the common stream, allowing more messages to be assigned to it, which decreases the transmission burden of the private streams while increasing that of the common stream. If the gain from the reduced private stream burden exceeds the loss from the heavier common stream load, the overall performance improves. As shown in the figure, this rate-splitting adjustment does not improve the overall performance at short blocklengths, so the proposed algorithm does not assign additional messages to the common stream, resulting in similar performance for the two $\lambda$ values. As the blocklength increases, the algorithm benefits from assigning more messages to the common stream, thereby enhancing the overall performance while inevitably reducing the common performance. These results demonstrate that adjusting the QoS constraint through the parameter $\lambda$ enables a flexible trade-off between the common and overall performance, indicating that an appropriate $\lambda$ can enhance the overall performance while preserving the superiority of the common stream.

\section{Conclusion}
\label{section VI}
In this work, we have investigated the AAoI performance of short-packet rate splitting in high-mobility autonomous driving scenarios. Considering system parameters including transmit power, vehicle velocity, blocklength, and the number of transmit antennas, we have derived closed-form AAoI expressions for both the common stream and the overall system. Leveraging these expressions, we have proposed a multi-start two-step SCA algorithm to optimize power allocation and rate splitting, with the QoS constraint regulating the rate-splitting strategy. Simulation results have demonstrated that the proposed RSMA scheme significantly improves the AAoI performance while guaranteeing the performance of the worst-case user and enabling ultra-low AAoI through the common stream, and achieves flexible adjustment between the common and overall performance.

\appendices
\section{Proof of Lemma~\ref{lemma1}}
\label{appendixA}
Based on (\ref{6}), the RV $X$ is approximated as
\begin{equation}
X \approx \rho^2\vert \mathbf{h}_k^H[m-1]\mathbf{p}_k\vert^2+(1-\rho^2)\vert \mathbf{e}_k^H[m]\mathbf{p}_k\vert^2.
\end{equation}
As mentioned earlier, $\vert\mathbf{h}_k^H[m-1]\mathbf{p}_k\vert^2 \!\sim\! \mathrm{Gamma}(N_t\!-\!K+1,1)$ and $\vert \mathbf{e}_k^H[m]\mathbf{p}_k\vert^2 \sim \mathrm{Gamma}(1,1)$. Therefore, we can further derive that $\rho^2\vert\mathbf{h}_k^H[m-1]\mathbf{p}_k\vert^2 \sim \mathrm{Gamma}(N_t-K+1,\rho^2)$ and $(1-\rho^2)\vert \mathbf{e}_k^H[m]\mathbf{p}_k\vert^2 \sim \mathrm{Gamma}(1,1-\rho^2)$. Substituting these two terms back into (\ref{8}) yields (\ref{9}), and the RV $X$ can be approximated by $\widetilde{X} \sim \mathrm{Gamma}(\widetilde{D}_1,\widetilde{\theta}_1)$.

\section{Proof of Lemma~\ref{lemma4}}
\label{appendixB}
We first define a new RV $M_{c,k}$ as
\begin{equation}
M_{c,k}=\frac{G_1}{P_n\xi_k\widetilde{\theta}_{2,k}G_{2,k}+1},
\end{equation}
where $G_1 = \vert\mathbf{h}_k^H[m]\mathbf{p}_c\vert^2 \sim \mathrm{Gamma}(1,1)$ and $G_{2,k} = \sum_{j\in\mathcal{K}}\alpha_j\vert\mathbf{h}_k^H[m]\mathbf{p}_j\vert^2 / \widetilde{\theta}_{2,k} = Y_k / \widetilde{\theta}_{2,k}$. From Lemma~\ref{lemma2}, the RV $Y_k$ can be approximated by $\widetilde{Y}_k \!\sim\! \mathrm{Gamma}(\widetilde{D}_{2,k},\widetilde{\theta}_{2,k})$, where $\widetilde{D}_{2,k}$ and $\widetilde{\theta}_{2,k}$ are given in (\ref{10}). Accordingly, the RV $G_{2,k}$ can be approximated by $\widetilde{G}_{2,k} \sim \mathrm{Gamma}(\widetilde{D}_{2,k},1)$. Under the assumption of independence between the numerator and denominator~\cite{mechanism_independent}, the RV $M_{c,k}$ can be approximated by
\begin{equation}
\widetilde{M}_{c,k}=\frac{G_1}{P_n\xi_k\widetilde{\theta}_{2,k}\widetilde{G}_{2,k}+1}. 
\end{equation}
Then, the CDF of $\widetilde{M}_{c,k}$ is given by
\begin{align}
\label{46}
&F_{\widetilde{M}_{c,k}}(y) = \mathbb{P} \Bigl( G_1\leq y(P_n\xi_k\widetilde{\theta}_{2,k}\widetilde{G}_{2,k}+1) \Bigr) \nonumber \\
&\quad= \int_0^\infty \mathbb{P} \Bigl( G_1\leq y(P_n\xi_k\widetilde{\theta}_{2,k}x+1) \Bigr)f_{\widetilde{G}_{2,k}}(x) \, \mathrm{d}x \nonumber \\
&\quad= \int_0^\infty \Bigl( 1-e^{-y(P_n\xi_k\widetilde{\theta}_{2,k}x+1)} \Bigr)\frac{x^{\widetilde{D}_{2,k}-1}e^{-x}}{\Gamma(\widetilde{D}_{2,k})} \, \mathrm{d}x \nonumber \\
&\quad= 1-\frac{e^{-y}}{\Gamma(\widetilde{D}_{2,k})}\int_0^\infty e^{-(yP_n\xi_k\widetilde{\theta}_{2,k}+1)x}x^{\widetilde{D}_{2,k}-1} \, \mathrm{d}x.
\end{align}
Applying a variable substitution $z=( yP_n\xi_k\widetilde{\theta}_{2,k}+1 )x$, (\ref{46}) is rewritten as
\begin{align}
\label{47}
F_{\widetilde{M}_{c,k}}(y) &= 1-\frac{e^{-y}}{(yP_n\xi_k\widetilde{\theta}_{2,k}+1)^{\widetilde{D}_{2,k}}}\int_0^\infty \frac{z^{\widetilde{D}_{2,k}-1}e^{-z}}{\Gamma(\widetilde{D}_{2,k})} \, \mathrm{d}z \nonumber \\
&= 1-\frac{e^{-y}}{(yP_n\xi_k\widetilde{\theta}_{2,k}+1)^{\widetilde{D}_{2,k}}}\int_0^\infty f_{\widetilde{G}_{2,k}}(z) \, \mathrm{d}z \nonumber \\
&= 1-\frac{e^{-y}}{(yP_n\xi_k\widetilde{\theta}_{2,k}+1)^{\widetilde{D}_{2,k}}}.
\end{align}
Since $\widetilde{\Gamma}_{c,k} = P_n\xi_k\alpha_c\widetilde{M}_{c,k}$, substituting it into (\ref{47}) and we can obtain (\ref{12}).

\section{Proof of Proposition~\ref{proposition1}}
\label{appendixC}
Based on (\ref{15}), the average BLER of the common stream can be approximated as
\begin{align}
&\varepsilon_{c,k} \approx \mathbb{E}\left[ L(\Gamma_{c,k}) \right] = \int_0^\infty\!\! L(x)f_{\Gamma_{c,k}}(x) \, \mathrm{d}x \nonumber\\
&= \int_0^\mu\!\! f_{\Gamma_{c,k}}(x) \, \mathrm{d}x + \int_\mu^\nu\!\! \left[ 0.5-\delta(x-\beta) \right]f_{\Gamma_{c,k}}(x) \, \mathrm{d}x \nonumber\\
&= (0.5\!-\!\delta\beta)F_{\Gamma_{c,k}}(\mu) + (0.5\!+\!\delta\beta)F_{\Gamma_{c,k}}(\nu) - \delta \!\!\int_\mu^\nu\!\!\! x \, \mathrm{d}F_{\Gamma_{c,k}}(x) \nonumber\\
&= (0.5+\delta\mu-\delta\beta)F_{\Gamma_{c,k}}(\mu) + (0.5+\delta\beta-\delta\nu)F_{\Gamma_{c,k}}(\nu) \nonumber\\
& \quad + \delta \!\!\int_\mu^\nu\!\! F_{\Gamma_{c,k}}(x) \, \mathrm{d}x.
\end{align}
As mentioned earlier, $\mu=\beta-\frac{1}{2\delta}$, and $\nu=\beta+\frac{1}{2\delta}$. It follows that $0.5+\delta\mu-\delta\beta=0.5+\delta\beta-\delta\nu=0$, therefore we obtain $\varepsilon_{c,k} \approx \delta\int_\mu^\nu F_{\Gamma_{c,k}}(x) \, \mathrm{d}x$, which can be further approximated as $\varepsilon_{c,k} \approx \delta(\nu-\mu)F_{\Gamma_{c,k}}(\frac{\mu+\nu}{2}) = F_{\Gamma_{c,k}}(2^{\frac{m_c}{n_c}}-1) \nonumber$ by using the midpoint rule. Substituting the approximate CDF of $\Gamma_{c,k}$ from (\ref{12}) into the above expression yields (\ref{16}).

\section{Proof of Lemma~\ref{lemma6}}
\label{appendixD}
Based on (\ref{19}), the AAoI can be further derived as
\begin{align}
\label{49}
\bar{\Delta} &= \lim_{\tau\to\infty}\frac{1}{\tau}\Bigl( Q_0 + \sum_{j=1}^{N(\tau)}Q_j + \Delta Q \Bigr) \nonumber\\
&= \lim_{\tau\to\infty}\frac{T^2+2\Delta Q}{2\tau}+\lim_{\tau\to\infty}\frac{\sum_{j=1}^{N(\tau)}Q_j}{\tau} \nonumber\\
&= \lim_{\tau\to\infty}\frac{N(\tau)}{\tau} \frac{\sum_{j=1}^{N(\tau)}Q_j}{N(\tau)}.
\end{align}
We define $W_j$ as the number of slots between the $(j-1)$-th and $j$-th successful receptions. Thus, the average time elapsed for each new successful reception is $T \mathbb{E}[W]$. According to the strong law of large numbers, as $\tau\to\infty$, the number of successful receptions per unit time converges to
\begin{equation}
\label{50}
\lim_{\tau\to\infty} \frac{N(\tau)}{\tau}=\frac{1}{T\mathbb{E}[W]}.
\end{equation}
Similarly, we can obtain
\begin{equation}
\label{51}
\lim_{\tau\to\infty} \frac{\sum_{j=1}^{N(\tau)}Q_j}{N(\tau)}=\mathbb{E}[Q],
\end{equation}
where $Q_j = T^2(W_j^2/2+W_j)$. By substituting (\ref{50}) and (\ref{51}) into (\ref{49}), we obtain
\begin{equation}
\label{52}
\bar{\Delta} = \frac{\mathbb{E}[Q]}{T\mathbb{E}[W]} = T \Bigl( \frac{\mathbb{E}[W^2]}{2\mathbb{E}[W]}+1 \Bigr).
\end{equation}
The reception failure probability is $\epsilon$, then the probability distribution of $W$ is given by
\begin{equation}
\mathbb{P}(W=n) = \epsilon^{n-1}(1-\epsilon), \; n\in\mathbb{Z}^+.
\end{equation}
Since $W$ follows a geometric distribution, we can obtain that $\mathbb{E}[W] = 1/(1-\epsilon)$ and $\mathbb{E}[W^2] = (1+\epsilon)/(1-\epsilon)^2$. Substituting the above into (\ref{52}) yields (\ref{20}).

\bibliographystyle{IEEEtran}
\bibliography{References}

\end{document}